\documentclass[envcountsame]{llncs}
\usepackage{times}
\usepackage{multicol}
\usepackage{amssymb}
\usepackage{amsmath}
\usepackage{booktabs,colortbl,pifont}
\usepackage{mdframed}
\usepackage{pgfplotstable,filecontents}
\pgfplotsset{compat=1.9}
\usepackage{graphicx}
\usepackage{xspace}
\usepackage{hyperref}
\usepackage{cleveref}
\usepackage{comment}

\makeatletter
\newcommand*{\bdiv}{%
  \nonscript\mskip-\medmuskip\mkern5mu%
  \mathbin{\operator@font div}\penalty900\mkern5mu%
  \nonscript\mskip-\medmuskip
}
\makeatother

\newcommand{\vars}{\ensuremath{X}\xspace}
\newcommand{\varss}[1]{\ensuremath{\vars_{#1}}\xspace}

\newcommand{\alive}{Alive\xspace}
\newcommand{\instcombine}{\emph{instcombine}\xspace}
\newcommand{\llvm}{LLVM\xspace}
\newcommand{\cpp}{C++\xspace}
\newcommand{\nsw}{\emph{nsw}\xspace}
\newcommand{\nuw}{\emph{nuw}\xspace}
\newcommand{\exact}{\emph{exact}\xspace}
\newcommand{\cvcfour}{CVC4\xspace}
\newcommand{\coq}{Coq\xspace}
\newcommand{\isabelle}{Isabelle\xspace}

\newcommand{\sygus}{SyGuS\xspace}
\newcommand{\zthree}{Z3\xspace}
\newcommand{\vampire}{Vampire\xspace}
\newcommand{\solidity}{Solidity\xspace}
\newcommand{\bi}{\begin{itemize}}
\newcommand{\ei}{\end{itemize}}
\newcommand{\be}{\begin{enumerate}}
\newcommand{\ee}{\end{enumerate}}
\newcommand{\bd}{\begin{description}}
\newcommand{\ed}{\end{description}}
\newcommand{\set}[1]{\left\{#1\right\}}
\newcommand{\brackets}[1]{\left[#1\right]}

\newcommand{\til}{,\ldots,}
\newcommand{\tonat}[1]{{\brackets{#1}}_{\mathbb{N}}}
\newcommand{\toint}[1]{{\brackets{#1}}_{\mathbb{Z}}}

\newcommand{\smtlib}{SMT-LIB~2\xspace}

\newcommand{\ite}[3]{\ensuremath{\mathrm{ite}(#1,#2,#3)}\xspace}

\newcommand{\xmark}{\ding{53}}%
\newcommand{\lmark}{\raisebox{-.2ex}{{\ding{213}}}}
\newcommand{\rmark}{\raisebox{.30ex}{{\rotatebox[origin=c]{-180}{\lmark}}}}

\newcommand{\ltr}{\lmark\xspace}
\newcommand{\rtl}{\rmark\xspace}
\newcommand{\both}{{\color{blue}{$\checkmark$}}\xspace}
\newcommand{\none}{{\color{red}\xmark}\xspace}
\newcommand{\ltrci}{\ensuremath{\ltr_{\ci}}\xspace}
\newcommand{\ltrnoci}{\ensuremath{\ltr_{\text{no}\,\ci}}\xspace}

\newcommand{\I}{\ensuremath{{\mathcal{I}}}\xspace}
\newcommand{\J}{\ensuremath{{\mathcal{J}}}\xspace}
\newcommand{\Is}{\ensuremath{{I}}\xspace}

\newcommand{\fv}[1]{\ensuremath{\mathrm{FV}\kern-.15em(#1)}\xspace}

\newcommand{\bv}[2]{\ensuremath{#1_{[#2]}}\xspace}

\renewcommand{\vec}[1]{\ensuremath{\boldsymbol{#1}}\xspace}

\newcommand{\fsbv}{\ensuremath{\mathrm{BV}}\xspace}

\newcommand{\ia}{\ensuremath{\mathrm{IA}}\xspace}

\newcommand{\ufia}{\ensuremath{\mathrm{UF}\ia}\xspace}
\newcommand{\sig}{\ensuremath{\Sigma}\xspace}
\newcommand{\sigs}{\ensuremath{\sig^s}\xspace}
\newcommand{\sigf}{\ensuremath{\sig^f}\xspace}

\newcommand{\sigfsbv}{\ensuremath{\sig_{\fsbv}}\xspace}

\newcommand{\sigia}{\ensuremath{\sig_{\ia}}\xspace}

\newcommand{\tfsbv}{\ensuremath{T_{\kern -.1em\fsbv}}\xspace}

\newcommand{\tia}{\ensuremath{T_{\ia}}\xspace}

\newcommand{\tufia}{\ensuremath{T_{\ufia}}\xspace}

\newcommand{\logicufnia}{{UFNIA}\xspace}

\newcommand{\sort}{\ensuremath{\sigma}\xspace}
\newcommand{\sorts}{\ensuremath{S}\xspace}
\newcommand{\sorti}{\ensuremath{\sort^\I}\xspace}
\newcommand{\sortbv}[1]{\ensuremath{\sort_{[#1]}}\xspace}
\newcommand{\Inte}{\ensuremath{\mathsf{Int}}\xspace}

\newcommand{\sub}[3]{\ensuremath{#1\set{#2 \mapsto #3}}\xspace}

\newcommand{\bool}{\ensuremath{\mathsf{Bool}}\xspace}

\newcommand{\sorteddomain}[2]{{#1}^{#2}}

\newcommand{\teq}{\ensuremath{\approx}\xspace}
\newcommand{\tneq}{\ensuremath{\not\teq}\xspace}

\newcommand{\teqf}[2]{\ensuremath{{#1 \kern-.15em\teq\kern-.15em #2}}\xspace}
\newcommand{\tneqf}[2]{\ensuremath{{#1 \kern-.15em\tneq\kern-.15em #2}}\xspace}

\newcommand{\op}{\ensuremath{\diamond}\xspace}

\newcommand{\fsbvtheory}{{\mathbb{\fsbv}}}

\newcommand{\bvsymbolf}[2]{{#1}^{#2\scaleto{\fsbvtheory}{4pt}}}
\newcommand{\bvsymbol}[1]{\bvsymbolf{#1}{}\xspace}
\newcommand{\nattheory}{{\mathbb{N}}\xspace}
\newcommand{\natsymbol}[1]{{#1}^{\scaleto{\nattheory}{4pt}}}

\newcommand{\traf}{\ensuremath{\mathcal{T}}\xspace}

\newcommand{\ic}{\ensuremath{\phi_\mathrm{c}}\xspace}
\newcommand{\ci}{\ensuremath{\alpha_\mathrm{c}}\xspace}

\newcommand{\lequiv}{\Leftrightarrow}
\newcommand{\rel}{\ensuremath{\bowtie}\xspace}

\newcommand{\bvaddf}{\ensuremath{\bvsymbolf{+}{}}\xspace}
\newcommand{\bvsubf}{\ensuremath{\bvsymbolf{-}{}}\xspace}
\newcommand{\bvandf}{\ensuremath{\mathrel{\bvsymbolf{\&}{}}}\xspace}
\newcommand{\bvashrf}{\ensuremath{\mathop{\bvsymbolf{>\kern-.3em>_a}{\kern -.6em}}}\xspace}
\newcommand{\bvconcatf}{\ensuremath{\bvsymbolf{\circ}{}}\xspace}
\newcommand{\bvlshrf}{\ensuremath{\bvsymbolf{\mathop{>\kern-.3em>}}{\kern -.1em}}\xspace}
\newcommand{\bvmulf}{\ensuremath{\bvsymbolf{\cdot}{}}\xspace}
\newcommand{\bvorf}{\ensuremath{\bvsymbolf{\mid}{}}\xspace}
\newcommand{\bvxorf}{\ensuremath{\bvsymbolf{\oplus}{}}\xspace}
\newcommand{\bvshlf}{\ensuremath{\bvsymbolf{\mathop{<\kern-.3em<}}{}}\xspace}
\newcommand{\bvudivf}{\ensuremath{\bvsymbolf{\bdiv}{}}\xspace}
\newcommand{\bvsgef}{\ensuremath{\bvsymbolf{\ge_\mathrm{s}}{\kern-.4em}}\xspace}
\newcommand{\bvsgtf}{\ensuremath{\bvsymbolf{>_\mathrm{s}}{\kern-.4em}}\xspace}
\newcommand{\bvslef}{\ensuremath{\bvsymbolf{\le_\mathrm{s}}{\kern-.4em}}\xspace}
\newcommand{\bvsltf}{\ensuremath{\bvsymbolf{<_\mathrm{s}}{\kern-.4em}}\xspace}
\newcommand{\bvugef}{\ensuremath{\bvsymbolf{\geq_\mathrm{u}}{\kern-.4em}}\xspace}
\newcommand{\bvugtf}{\ensuremath{\bvsymbolf{>_\mathrm{u}}{\kern-.4em}}\xspace}
\newcommand{\bvultf}{\ensuremath{\bvsymbolf{<_\mathrm{u}}{\kern-.4em}}\xspace}
\newcommand{\bvulef}{\ensuremath{\bvsymbolf{\leq_\mathrm{u}}{\kern-.4em}}\xspace}
\newcommand{\bvuremf}{\ensuremath{\bvsymbolf{\bmod}{}}\xspace}
\newcommand{\bvnegf}{\ensuremath{\bvsymbolf{-}{}}\xspace}
\newcommand{\bvnotf}{\ensuremath{\bvsymbolf{{\sim}\,}{}}\xspace}

\newcommand{\bvmaxsf}{\ensuremath{\bvsymbolf{\mathrm{max}_\mathrm{s}}{\kern-.4em}}\xspace}
\newcommand{\bvminsf}{\ensuremath{\bvsymbolf{\mathrm{min}_\mathrm{s}}{\kern-.4em}}\xspace}

\newcommand{\bvadd}[2]{\ensuremath{#1\, \bvaddf\kern -.1em #2}\xspace}
\newcommand{\bvand}[2]{\ensuremath{#1 \bvandf\kern -.35em #2}\xspace}
\newcommand{\bvashr}[2]{\ensuremath{#1\, \bvashrf\kern -.1em #2}\xspace}
\newcommand{\bvconcat}[2]{\ensuremath{#1 \bvconcatf\kern -.1em #2}\xspace}
\newcommand{\bvextract}[3]{\ensuremath{#1[#2:#3]^{\scaleto{\fsbvtheory}{4pt}}}\xspace}
\newcommand{\bvlshr}[2]{\ensuremath{#1\, \bvlshrf\kern -.1em #2}\xspace}
\newcommand{\bvmul}[2]{\ensuremath{#1\, \bvmulf\kern -.1em #2}\xspace}
\newcommand{\bvneg}[1]{\ensuremath{\bvnegf\kern -.1em#1}\xspace}
\newcommand{\bvnot}[1]{\ensuremath{\bvnotf\kern -.1em#1}\xspace}
\newcommand{\bvor}[2]{\ensuremath{#1\, \bvorf\kern -.1em #2}\xspace}
\newcommand{\bvxor}[2]{\ensuremath{#1 \bvxorf\kern -.1em #2}\xspace}

\newcommand{\bvshl}[2]{\ensuremath{#1\, \bvshlf\kern -.1em #2}\xspace}

\newcommand{\bvslt}[2]{\ensuremath{#1\, \bvsltf #2}\xspace}
\newcommand{\bvsub}[2]{\ensuremath{#1 \bvsubf\kern -.1em #2}\xspace}
\newcommand{\bvudiv}[2]{\ensuremath{#1\, \bvudivf\kern -.1em #2}\xspace}

\newcommand{\bvule}[2]{\ensuremath{#1\, \bvulef #2}\xspace}

\newcommand{\bvurem}[2]{\ensuremath{#1\, \bvuremf\kern -.1em #2}\xspace}
\newcommand{\dist}[2]{\ensuremath{#1 \tneq #2}\xspace}
\newcommand{\equal}[2]{\ensuremath{#1 \teq #2}\xspace}

\newcommand{\bvmins}[1]{\ensuremath{\bvminsf_{[#1]}}\xspace}
\newcommand{\bvmaxs}[1]{\ensuremath{\bvmaxsf_{[#1]}}\xspace}

\newcommand{\bvnsymbolf}[2]{#1\xspace}
\newcommand{\bvnaddf}{\ensuremath{\bvnsymbolf{+}{}}\xspace}
\newcommand{\bvnsubf}{\ensuremath{\bvnsymbolf{-}{}}\xspace}
\newcommand{\bvnandf}{\ensuremath{\mathrel{\bvnsymbolf{\&}{}}}\xspace}
\newcommand{\bvnashrf}{\ensuremath{\mathop{\bvnsymbolf{>\kern-.3em>_a}{\kern -.6em}}}\xspace}
\newcommand{\bvnconcatf}{\ensuremath{\bvnsymbolf{\circ}{}}\xspace}
\newcommand{\bvnlshrf}{\ensuremath{\bvnsymbolf{\mathop{>\kern-.3em>}}{\kern -.1em}}\xspace}
\newcommand{\bvnmulf}{\ensuremath{\bvnsymbolf{\cdot}{}}\xspace}
\newcommand{\bvnorf}{\ensuremath{\bvnsymbolf{\mid}{}}\xspace}
\newcommand{\bvnxorf}{\ensuremath{\bvnsymbolf{\oplus}{}}\xspace}
\newcommand{\bvnshlf}{\ensuremath{\bvnsymbolf{\mathop{<\kern-.3em<}}{}}\xspace}
\newcommand{\bvnudivf}{\ensuremath{\bvnsymbolf{\bdiv}{}}\xspace}
\newcommand{\bvnsgef}{\ensuremath{\bvnsymbolf{\ge_s}{\kern-.4em}}\xspace}
\newcommand{\bvnsgtf}{\ensuremath{\bvnsymbolf{>_s}{\kern-.4em}}\xspace}
\newcommand{\bvnslef}{\ensuremath{\bvnsymbolf{\le_s}{\kern-.4em}}\xspace}
\newcommand{\bvnsltf}{\ensuremath{\bvnsymbolf{<_s}{\kern-.4em}}\xspace}
\newcommand{\bvnugef}{\ensuremath{\bvnsymbolf{\geq_u}{\kern-.4em}}\xspace}
\newcommand{\bvnugtf}{\ensuremath{\bvnsymbolf{>_u}{\kern-.4em}}\xspace}
\newcommand{\bvnultf}{\ensuremath{\bvnsymbolf{<_u}{\kern-.4em}}\xspace}
\newcommand{\bvnulef}{\ensuremath{\bvnsymbolf{\leq_u}{\kern-.4em}}\xspace}
\newcommand{\bvnuremf}{\ensuremath{\bvnsymbolf{\bmod}{}}\xspace}
\newcommand{\bvnnegf}{\ensuremath{\bvnsymbolf{-}{}}\xspace}
\newcommand{\bvnnotf}{\ensuremath{\bvnsymbolf{{\sim}\,}{}}\xspace}

\newcommand{\bvnmaxsf}{\ensuremath{\bvnsymbolf{\text{max}_s}{\kern-.4em}}\xspace}
\newcommand{\bvnminsf}{\ensuremath{\bvnsymbolf{\text{min}_s}{\kern-.4em}}\xspace}

\newcommand{\bvnadd}[2]{\ensuremath{#1\, \bvnaddf\kern -.1em #2}\xspace}
\newcommand{\bvnand}[2]{\ensuremath{#1 \bvnandf\kern -.15em #2}\xspace}
\newcommand{\bvnashr}[2]{\ensuremath{#1\, \bvnashrf\kern -.1em #2}\xspace}
\newcommand{\bvnconcat}[2]{\ensuremath{#1 \bvnconcatf\kern -.1em #2}\xspace}

\newcommand{\bvnlshr}[2]{\ensuremath{#1\, \bvnlshrf\kern -.1em #2}\xspace}
\newcommand{\bvnmul}[2]{\ensuremath{#1\, \bvnmulf\kern -.1em #2}\xspace}
\newcommand{\bvnneg}[1]{\ensuremath{\bvnnegf\kern -.1em#1}\xspace}
\newcommand{\bvnnot}[1]{\ensuremath{\bvnnotf\kern -.1em#1}\xspace}
\newcommand{\bvnor}[2]{\ensuremath{#1\, \bvnorf\kern -.1em #2}\xspace}
\newcommand{\bvnxor}[2]{\ensuremath{#1 \bvnxorf\kern -.1em #2}\xspace}
\newcommand{\bvnsge}[2]{\ensuremath{#1\, \bvnsgef #2}\xspace}
\newcommand{\bvnsgt}[2]{\ensuremath{#1\, \bvnsgtf #2}\xspace}
\newcommand{\bvnshl}[2]{\ensuremath{#1\, \bvnshlf\kern -.1em #2}\xspace}
\newcommand{\bvnsle}[2]{\ensuremath{#1\, \bvnslef #2}\xspace}
\newcommand{\bvnslt}[2]{\ensuremath{#1\, \bvnsltf #2}\xspace}
\newcommand{\bvnsub}[2]{\ensuremath{#1 \bvnsubf\kern -.1em #2}\xspace}
\newcommand{\bvnudiv}[2]{\ensuremath{#1\, \bvnudivf\kern -.1em #2}\xspace}
\newcommand{\bvnuge}[2]{\ensuremath{#1\, \bvnugef #2}\xspace}
\newcommand{\bvnugt}[2]{\ensuremath{#1\, \bvnugtf #2}\xspace}
\newcommand{\bvnule}[2]{\ensuremath{#1\, \bvnulef #2}\xspace}
\newcommand{\bvnult}[2]{\ensuremath{#1\, \bvnultf #2}\xspace}
\newcommand{\bvnurem}[2]{\ensuremath{#1\, \bvnuremf\kern -.1em #2}\xspace}

\newcommand{\true}{\ensuremath{\top}\xspace}

\renewcommand{\v}{\vee}
\newcommand{\w}{\wedge}
\newcommand{\Ra}{\Rightarrow}

\newcommand{\bvselect}[2]{{#1}\brackets{#2}}

\newcommand{\twotothef}{\ensuremath{\mathsf{pow2}}\xspace}
\newcommand{\twotothe}[1]{\ensuremath{\twotothef(#1)}\xspace}

\newcommand{\maptointvar}{\ensuremath{\chi}\xspace}

\newcommand{\intaddf}{\ensuremath{\natsymbol{+}}\xspace}
\newcommand{\intsubf}{\ensuremath{\natsymbol{-}}\xspace}
\newcommand{\intandf}{\ensuremath{\mathrel{\natsymbol{\&}\kern-.3em}}\xspace}
\newcommand{\intashrf}{\ensuremath{\mathop{\natsymbol{>\kern-.3em>_a}}}\xspace}
\newcommand{\intconcatf}{\ensuremath{\natsymbol{\circ}}\xspace}
\newcommand{\intlshrf}{\ensuremath{\natsymbol{\mathop{>\kern-.3em>}}}\xspace}
\newcommand{\intmulf}{\ensuremath{\natsymbol{\cdot}}\xspace}
\newcommand{\intorf}{\ensuremath{\natsymbol{\mid}}\xspace}
\newcommand{\intxorf}{\ensuremath{\natsymbol{\oplus}}\xspace}

\newcommand{\intshlf}{\ensuremath{\natsymbol{\mathop{<\kern-.3em<}}}\xspace}

\newcommand{\intsltf}{\ensuremath{\natsymbol{<_s}}\xspace}
\newcommand{\intudivf}{\ensuremath{\natsymbol{\bdiv}}\xspace}
\newcommand{\intugef}{\ensuremath{\natsymbol{\geq_u}}\xspace}

\newcommand{\inturemf}{\ensuremath{\natsymbol{\bmod}}\xspace}
\newcommand{\intnegf}{\ensuremath{\natsymbol{-}}\xspace}
\newcommand{\intnotf}{\ensuremath{\natsymbol{{\sim}\,}}\xspace}
\newcommand{\intmaxf}{\ensuremath{\natsymbol{\max}}\xspace}

\newcommand{\intadd}[3]{\intaddf(#1,#2,#3)}
\newcommand{\intmul}[3]{\intmulf(#1,#2,#3)}
\newcommand{\inturem}[3]{\inturemf(#1,#2,#3)}
\newcommand{\intudiv}[3]{\intudivf(#1,#2,#3)}
\newcommand{\intsub}[3]{\intsubf(#1,#2,#3)}
\newcommand{\intand}[3]{\intandf(#1,#2,#3)}

\newcommand{\intxor}[3]{\intxorf(#1,#2,#3)}
\newcommand{\intnot}[2]{\intnotf(#1, #2)}
\newcommand{\intneg}[2]{\intnegf(#1, #2)}

\newcommand{\intshl}[3]{\intshlf(#1,#2,#3)}
\newcommand{\intlshr}[3]{\intlshrf({#1},{#2},{#3})}

\newcommand{\intconcat}[3]{\intconcatf(#1,#2,#3)}

\newcommand{\intslt}[3]{\intsltf(#1, #2, #3)}

\newcommand{\intmax}[1]{\intmaxf_{#1}}

\newcommand{\utsf}{\ensuremath{\mathsf{uts}}\xspace}
\newcommand{\uts}[2]{\ensuremath{\utsf_{#1}(#2)}\xspace}

\newcommand{\hextractf}{\ensuremath{\mathsf{ex}}\xspace}

\newcommand{\hselect}[3]{\ensuremath{\hextractf_{#2}(#3)}\xspace}

\renewcommand{\mod}{\ensuremath{\bmod}\xspace}
\renewcommand{\div}{\ensuremath{\bdiv}\xspace}

\newcommand{\full}{\ensuremath{\mathrm{full}}\xspace}
\renewcommand{\partial}{\ensuremath{\mathrm{partial}}\xspace}
\newcommand{\qf}{\ensuremath{\mathrm{qf}}\xspace}
\newcommand{\combined}{\ensuremath{\mathrm{combined}}\xspace}

\newcommand{\aaxiom}[1]{\ensuremath{\mathrm{AX}_{\kern-.2em#1}}\xspace}
\newcommand{\axiom}[2]{\ensuremath{\mathrm{AX}_{#1}^{#2}}\xspace}
\newcommand{\abs}[1]{\ensuremath{\lvert{#1}\lvert}\xspace}

\newcommand{\trfullf}{\traf_{\full}\xspace}
\newcommand{\trcombinedf}{\traf_{\combined}\xspace}
\newcommand{\trpartialf}{\traf_{\partial}\xspace}
\newcommand{\trqff}{\traf_{\qf}\xspace}

\newcommand{\trfull}[1]{\trfullf(#1)\xspace}
\newcommand{\trcombined}[1]{\trcombinedf(#1)\xspace}
\newcommand{\trpartial}[1]{\trpartialf(#1)\xspace}
\newcommand{\trqf}[1]{\trqff(#1)\xspace}

\newcommand{\parvars}{X^\ast\xspace}
\newcommand{\parconsts}{Z^\ast\xspace}
\newcommand{\parvar}{\mathsf{x}\xspace}
\newcommand{\parconst}{\mathsf{z}\xspace}
\newcommand{\instpar}[3]{#1|_{#2[#3]}\xspace}

\newcommand{\runconvert}{\textsc{Conv}\xspace}
\newcommand{\runelim}{\textsc{Elim}\xspace}

\newcommand{\omegabw}{\omega^b\xspace}
\newcommand{\omegaval}{\omega^N\xspace}

\newcommand{\bvi}[1]{{#1}\scaleto{\fsbvtheory}{4pt}}
\newcommand{\ufiai}[1]{{#1}\scaleto{\nattheory}{4pt}}

\usepackage{mdframed}
\usepackage{multirow}
\usepackage{scalerel}
\usepackage{xcolor}

\newsavebox{\algorithmbox}
\begin{document}

\newcommand\mytitle{Towards Bit-Width-Independent Proofs \\ in SMT Solvers}


\title{\mytitle\thanks{This work was supported in part by DARPA (awards N66001-18-C-4012 and FA8650-18-2-7861),
ONR (award N68335-17-C-0558),
NSF (award 1656926),
and the Stanford Center for Blockchain Research.
}
}

\hypersetup{
 pdfborder={0 0 0},
 colorlinks=true,
 linkcolor=blue,
 urlcolor=blue,
 citecolor=blue,
 pdfauthor={Yoni Zohar et al.},
 pdftitle=\mytitle
}
\titlerunning{\mytitle}

 \author{%
   Aina Niemetz\inst{1}
   \and
   Mathias Preiner\inst{1}
   \and
   Andrew Reynolds\inst{2}
   \and
   Yoni Zohar\inst{1}
   \and\\
   Clark Barrett\inst{1}
   \and
   Cesare Tinelli\inst{2}
}

%
\institute{{Stanford University, Stanford, USA} \and
{The University of Iowa, Iowa City, USA}}
\maketitle              
\begin{abstract}
Many SMT solvers implement efficient 
SAT-based procedures for solving fixed-size
bit-vector formulas.
These approaches, however, cannot be used directly
to reason about bit-vectors of symbolic bit-width.
To address this shortcoming,
we propose 
a translation from bit-vector formulas with parametric bit-width
to formulas in a logic supported by SMT solvers that includes
non-linear integer arithmetic, uninterpreted functions, and
universal quantification.
While this logic is undecidable, this approach can still solve many formulas by
capitalizing on
advances in SMT solving for non-linear arithmetic and universally
quantified formulas.
We provide several case studies in which we have applied
this approach with promising results,
including the bit-width independent verification 
of invertibility conditions, compiler optimizations,
and bit-vector rewrites.
\end{abstract}

\section{Introduction}

Satisfiability Modulo Theories (SMT) solving for the theory of fixed-size bit-vectors
has received a lot of interest in recent years.
Many applications rely on bit-precise reasoning as provided by SMT solvers, and
the number of solvers that participate in the corresponding divisions of the
annual SMT competition is high and increasing.
Although theoretically difficult (e.g.,~\cite{DBLP:journals/mst/KovasznaiFB16}),
bit-vector solvers are in practice highly efficient
and typically implement SAT-based procedures.
%
Reasoning about fixed-size bit-vectors suffices for many applications.
In hardware verification,
the size of a circuit is usually known in advance,
and in software verification,
machine integers are treated as fixed-size bit-vectors, where the width
depends on the underlying architecture.
Current solving approaches, however, do not generalize
beyond this limitation,
i.e., they cannot reason about
parametric circuits or
machine integers of arbitrary size.
This is a serious limitation when one wants to prove properties
that are bit-width independent.
Further, when reasoning about machine integers of a fixed but large size,
as employed, for example, in smart contract languages such as \solidity \cite{Solidity},
current approaches do not perform as well
in the presence of expensive operations such as multiplication \cite{DBLP:series/txtcs/KroeningS16}.

To address this limitation we propose a general method 
for reasoning about bit-vector formulas with parametric bit-width.
The essence of the method is
to replace the translation from fixed-size bit-vectors to propositional logic
(which is at the core of state-of-the-art bit-vector solvers) 
with a translation to the quantified theories of integer arithmetic
and uninterpreted functions.
We obtain a fully automated verification process
by capitalizing on
recent advances in SMT solving for these theories.

The reliability of our approach depends on the correctness
of the SMT solvers in use.
Interactive theorem provers, or proof assistants,
such as \isabelle and \coq~\cite{nipkow2002isabelle,coqref},
on the other hand,
target applications where trust is of higher importance than automation,
although
substantial progress towards increasing the latter
has been made in recent years~\cite{DBLP:journals/jar/BlanchetteBP13}.
Our long-term goal is
an efficient automated framework for proving bit-width independent properties
within a trusted proof assistant,
which requires
both a formalization of such properties in the language of the proof assistant
and the development of efficient automated techniques
to reason about these properties.
This work shows that state-of-the-art SMT solving combined with 
our encoding techniques make the latter feasible.
The next steps towards this goal are described
in the final section of this paper.


Translating a formula from the theory of fixed-size bit-vectors
to the theory of integer arithmetic
is not straightforward.
This is due to the fact that the semantics of bit-vector operators are defined
modulo the bit-width $n$, which must be expressed using exponentiation terms
$2^n$.
Most SMT solvers, however, do not support unrestricted exponentiation.
Furthermore, operators such as bit-wise $\mathit{and}$ and $\mathit{or}$
do not have a natural representation in integer arithmetic.
%
While they are definable in the theory of integer arithmetic using
$\beta$-function encodings (e.g.,~\cite{enderton2001mathematical}), such a
translation is expensive
as it requires an encoding of sequences into natural numbers.
Instead,
we introduce an uninterpreted function (UF) for each of the problematic operators
and axiomatize them with quantified formulas,
which shifts some of the burden from arithmetic to UF reasoning.
We consider two alternative axiomatizations:
a complete one relaying on induction,
and a partial (hand-crafted) one that can be understood as an under-approximation.

To evaluate the potential of our approach, we examine three case studies 
that arise from real applications where reasoning about bit-width independent properties
is essential.
Niemetz et al.~\cite{invcav18} defined invertibility conditions for bit-vector
operators, which they then used to solve quantified bit-vector formulas.
However, correctness of the conditions was only checked for specific bit-widths:
from $1$ to $65$.
As a first case study, we consider the bit-width independent verification
of these invertibility conditions,
which~\cite{invcav18} left to future work.
As a second case study,
we examine the bit-width independent verification of compiler optimizations
in \llvm.
For that, we use the \alive tool~\cite{alive15}, which
generates verification conditions for 
such optimizations in the theory of fixed-size bit-vectors.
Proving the correctness of these optimizations 
for arbitrary bit-widths would
ensure their correctness for any language and underlying architecture rather than specific ones.
As a third case study,
we consider the bit-width independent verification of rewrite rules for
the theory of fixed-size bit-vectors.
SMT solvers for this theory heavily rely on such rules to simplify the input.
Verifying their correctness is essential
and is typically done by hand, which is both tedious and error-prone.

To summarize, this paper makes the following contributions.
\begin{itemize}
  \item In \Cref{sec:encodings},
    we study complete and incomplete encodings of bit-vector formulas with
    parametric bit-width into integer arithmetic.
  \item In \Cref{sec:casestudy},
    we evaluate the effectiveness of both encodings in three case studies.
  \item As part of the invertibility conditions case study, 
    we introduce {\em conditional inverses} for
    bit-vector constraints, thus augmenting
    \cite{invcav18}
    with concrete parametric solutions.
\end{itemize}

\begin{table}[t]
  \centering
{%
  \renewcommand{\arraystretch}{1.2}%
  \begin{tabular}{l@{\hspace{1.7em}}l@{\hspace{1.7em}}l}
    \hline
    \textbf{Symbol} & \textbf{SMT-LIB Syntax} & \textbf{Sort} \\
    \hline

    \teq, \tneq
    & =, distinct
    & $\sort_{[n]} \times \sort_{[n]} \to \bool$
    \\
    \bvultf, \bvugtf, \bvsltf, \bvsgtf
    & bvult, bvugt, bvslt, bvsgt
    & $\sort_{[n]} \times \sort_{[n]} \to \bool$
    \\
    \bvulef, \bvugef, \bvslef, \bvsgef
    & bvule, bvuge, bvsle, bvsge
    & $\sort_{[n]} \times \sort_{[n]} \to \bool$
    \\
    \bvnotf, \bvnegf
    & bvnot, bvneg
    & $\sort_{[n]} \to \sort_{[n]}$
    \\
    \bvandf, \bvorf, \bvxorf
    & bvand, bvor, bvxor
    & $\sort_{[n]} \times \sort_{[n]} \to \sort_{[n]}$
    \\
    \bvshlf, \bvlshrf, \bvashrf
    & bvshl, bvlshr, bvashr
    & $\sort_{[n]} \times \sort_{[n]} \to \sort_{[n]}$
    \\
    \bvaddf, \bvmulf, \bvuremf, \bvudivf
    & bvadd, bvmul, bvurem, bvudiv
    & $\sort_{[n]} \times \sort_{[n]} \to \sort_{[n]}$
    \\
    \bvextract{}{u}{l}
    & extract ($0 \le l \le u < n$)
    & $\sort_{[n]} \to \sort_{[u-l+1]}$ 
    \\
    \bvconcatf
    & concatenation
    & $\sort_{[n]} \times \sort_{[m]} \to \sort_{[n+m]}$ 
    \\
    \hline
  \end{tabular}%
}

  \vspace{1ex}
  \caption{Considered bit-vector operators with \smtlib syntax.}
  \label{tab:bvops}
  \vspace*{-4ex}
\end{table}

\paragraph{Related Work}
Bit-width independent bit-vector formulas were studied by Picora 
\cite{PichoraPhd}, 
who introduced a formal language for bit-vectors of parametric width, 
along with a semantics and a decision procedure.
The language we use here is a simplified variant of that language.
A unification-based algorithm for bit-vectors of symbolic lengths is discussed by Bj{\o}rner and Picora~\cite{pichorabjorner98}.
Bit-width independent formulas are related to parametric Boolean functions and circuits. 
An inductive approach 
for reasoning about such formalisms was
developed by Gupta and Fisher~\cite{Gupta:1993:RSM:259794.259827,10.1007/3-540-56922-7_3}
by considering a Boolean function for the base case of a circuit and another one for its inductive step.
Reasoning about equivalence of such circuits can be embedded
in the framework of~\cite{PichoraPhd}.

\section{Preliminaries}
\label{sec:prelim}
  We briefly review 
  the usual notions and terminology of many-sorted first-order logic
  with equality (denoted by $\teq$).
  See~\cite{enderton2001mathematical,10.1007/978-3-540-30227-8_53} for more
  detailed information.
  Let \sorts be a set of \emph{sort symbols},
  and for every sort ${\sort \in \sorts}$,
  let~\varss{\sort} be an infinite set of \emph{variables of sort \sort}.
  We assume that sets \varss{\sort} are pairwise disjoint
  and define \vars as the union of sets \varss{\sort}.
  A \emph{signature} \sig 
  consists of
  a set
  $\sigs\!\subseteq \sorts$
  of sort symbols and
  a set
  $\sigf$
  of function symbols.
    Arities of function symbols are defined in the usual way.
  Constants are treated as $0$-ary functions.
  We assume that \sig includes a Boolean sort \bool
  and the Boolean constants $\top$ (true) and $\bot$ (false).  Functions
  returning \bool are also called \emph{predicates}.

  We assume the usual definitions of
  well-sorted terms, literals, and formulas, and refer to them as
   \sig-terms, \sig-literals, and \sig-formulas, respectively.
%
  %
  %
  %
  We define $\vec{x} = (x_1, ..., x_n)$ as a tuple of variables
  and write $Q\vec{x}\varphi$ with $Q \in \{ \forall, \exists \}$
  for a \emph{quantified} formula $Qx_1 \cdots Qx_n \varphi$.
  %
  %
  For a \sig-term or \sig-formula $e$,
  we denote the \emph{free variables} of $e$ (defined as usual) as \fv{e} and
  use $e[\vec{x}]$ to denote that the variables in \vec{x} occur free in $e$.
  For a tuple of \sig-terms $\vec{t} = (t_1, ..., t_n)$ and 
  a tuple of $\sig$-variables $\vec{x} = (x_{1}\til x_{n})$,
  we write $\sub{e}{\vec{x}}{\vec{t}}$ for the term or formula obtained from $e$
  by simultaneously replacing each occurrence of $x_i$ in $e$ by $t_i$.

  A \emph{\sig-interpretation} $\I$ maps:
  each $\sort \in \sigs$ to a distinct non-empty set of values \sorti
  (the \emph{domain} of \sort in \I);
  each $x \in \varss{\sort}$ to an element ${x^\I \in \sorti}$;
  and each $f^{\sort_1 \cdots \sort_n \sort} \in \sigf$ to a total function
  $f^\I\!\!: \sort_1^\I \times ... \times \sort_n^\I \to \sorti$ if $n > 0$,
  and to an element in \sorti if $n = 0$.
  %
  We use the usual inductive definition of a satisfiability relation $\models$
  between \sig-interpretations and \sig-formulas.

  %
%

  A \emph{theory}~$T$ is a pair $(\sig, \Is)$,
  where \sig is a signature
  and \Is is a non-empty class of \sig-interpretations
  that is closed under variable reassignment, i.e.,
  if interpretation $\I'$ only differs from an $\I\in\Is$
  in how it interprets variables, then also ${\I'\in\Is}$.
  A \sig-formula $\varphi$ is \emph{$T$-satisfiable}
  (resp.~\emph{$T$-unsatisfiable})
  if it is satisfied by some (resp.~no) interpretation in \Is;
  it is \emph{$T$-valid} if it is satisfied by all interpretations in \Is.
  We will sometimes omit $T$ when the theory is understood from context.

  %
  The theory $\tfsbv = (\sigfsbv, \Is_{\fsbv})$
  of fixed-size bit-vectors
  as defined in the \smtlib standard~\cite{SMTLib2010}
  consists of the class of interpretations $\Is_{\fsbv}$
  and signature $\sigfsbv$, which includes
  a unique sort for each positive 
  integer~$n$ (representing the bit-vector width),
  denoted here as $\sortbv{n}$.
  %
  %
  For a given positive integer $n$,  the domain $\sorteddomain{\sortbv{n}}{\I}$ of
  sort $\sortbv{n}$ in $\I$ is the set of all bit-vectors of size $n$.
  We assume that \sigfsbv includes all \emph{bit-vector constants}
  of sort $\sortbv{n}$ for each $n$, represented as bit-strings.
  However, to simplify the notation we will sometimes denote them by the corresponding
  natural number in $\{0, \ldots, 2^{n}-1\}$.
  All interpretations $\I \in \Is_{\fsbv}$ are identical except for the value
  they assign to variables.
  They interpret sort and function symbols as specified in \smtlib.
  All function symbols (of non-zero arity) in $\sigf_\mathrm{BV}$ are
  overloaded for every $\sortbv{n}\! \in \sigs_\mathrm{BV}$.
  We denote
  a \sigfsbv-term (or \emph{bit-vector term})~$t$ of width $n$ as \bv{t}{n}
  when we want to specify its bit-width explicitly.
  %
  We refer to the $i$-th bit of $t_{[n]}$ as $t[i]$ with $0\leq i < n$.
  We interpret $t[0]$ as the least significant bit (LSB),
  and $t[n-1]$ as the most significant bit (MSB),
  and denote bit ranges over $k$ from index $j$ down to $i$ as
  $t[j:i]$.
  The unsigned interpretation of a bit-vector $t_{[n]}$ as a natural number is given by
  $\tonat{t} = \Sigma_{i=0}^{n-1}\bvselect{t}{i}\cdot 2^i$, and its signed interpretation as
  an integer is given by $\toint{t} =
  -\bvselect{t}{n-1}\cdot 2^{n-1}+\tonat{\bvextract{t}{n-2}{0}}$.

  %
  %

  Without loss of generality,
  we consider a restricted set of bit-vector function and predicate symbols
  (or \emph{bit-vector operators})
  as listed in \Cref{tab:bvops}.
  The selection of operators in this set is arbitrary
  but complete in the sense that
  it suffices to express
  all bit-vector operators defined in \smtlib.
    We use $\bvmaxs{k}$ ($\bvmins{k}$)
  for the \emph{maximum} or \emph{minimum signed value} of width $k$,
  e.g., 
  $\bvmaxs{4}=0111$ and $\bvmins{4}=1000$.
  %

  The theory $\tia= (\sigia, \Is_{\ia})$ of integer arithmetic is also defined as
  in the \smtlib standard.
  The signature $\sigia$ includes a single sort  $\Inte$,
  function and predicate symbols
  $\set{+, -, \cdot, \div,\allowbreak \mod, \abs{...}, <, \leq, >, \geq}$,
  and
  a constant symbol for every integer value.
  We further extend \sigia to include
  exponentiation, denoted in the usual way as $a^b$.
  All interpretations $\I\in \Is_{\ia}$ are identical except for the values
  they assign to variables.
We write $\tufia$ to denote the (combined) theory 
of uninterpreted functions with integer arithmetic.
Its signature is the union of the signature of $\tia$
with a signature containing a set of (freely interpreted) function symbols,
called \emph{uninterpreted functions}.
  
\subsection{Parametric Bit-Vector Formulas}
\label{parabitvecform}
We are interested in reasoning about (classes of) $\sigfsbv$-formulas
that hold independently of the sorts assigned to their variables or terms.
We formalize the notion of parametric $\sigfsbv$-formulas in the following.

We fix two sets $\parvars$ and $\parconsts$ of variable and constant
symbols, respectively, of bit-vector sort of undetermined bit-width.
The bit-width is provided by the first component of a separate function pair 
$\omega = (\omegabw, \omegaval)$ which maps symbols 
$\parvar \in \parvars \cup \parconsts$ to $\sigia$-terms.
We refer to 
$\omegabw(\parvar)$
as the \emph{symbolic bit-width} assigned by $\omega$ to $\parvar$.
The second component of $\omega$ is a map $\omegaval$ from symbols
$\parconst \in \parconsts$ to $\sigia$-terms.
We call $\omegaval(\parconst)$ the \emph{symbolic value}
assigned by $\omega$ to $\parconst$.
Let $\vec{v}=\fv{\omega}$ be the set of
(integer) free variables occurring in the range of either $\omegabw$ or $\omegaval$.
We say that $\omega$ is \emph{admissible} if
for every interpretation $\I\in \Is_{\ia}$ that interprets each variable in 
$\vec{v}$ as a
positive integer, and for every $\parvar \in \parvars \cup \parconsts$,
$\I$ also interprets $\omegabw(\parvar)$ as a positive integer.

Let $\varphi$ be a formula built from the function symbols of $\sigfsbv$
and $\parvars \cup \parconsts$, ignoring their sorts.
We refer to $\varphi$ as a \emph{parametric $\sigfsbv$-formula}.
One can interpret $\varphi$ as a class of fixed-size 
bit-vector formulas as follows.
For each symbol $\parvar \in \parvars$ and 
integer $n>0$, 
we associate a unique variable $x_n$
of (fixed) bit-vector sort $\sortbv{n}$.
Given an admissible $\omega$ with $\vec{v}=\fv{\omega}$
and an interpretation $\I$ that maps each variable in $\vec{v}$ to a positive integer,
let $\instpar{\varphi}{\omega}{\I}$
be the result of replacing 
all symbols $\parvar \in \parvars$ in $\varphi$ by the
corresponding bit-vector variable $x_k$
and all symbols $\parvar \in \parconsts$ in $\varphi$ by 
the bit-vector constant of sort $\sortbv{k}$corresponding to 
$\omegaval( \parvar )^\I \mod 2^k$, where
in both cases $k$ is the value of 
$\omegabw(\parvar)^\I$.
We say a formula $\varphi$ is \emph{well sorted under $\omega$}
if $\omega$ is admissible and
$\instpar{\varphi}{\omega}{\I}$ is a well-sorted $\sigfsbv$-formula
for all $\I$ that map variables in $\vec{v}$ to positive integers.

\begin{example}
  \label{ex:parex}
Let $\parvars$ be the set $\{ \parvar \}$ and $\parconsts$ be the set 
$\{ \parconst_0, \parconst_1 \}$, where $\omegaval(\parconst_0) = 0$ and $\omegaval(\parconst_1) = 1$.
Let $\varphi$
be the formula $\bvadd{(\bvadd{\parvar}{\parvar})}{\parconst_1} \tneq \parconst_0$.
We have that $\varphi$ is well sorted under $(\omegabw, \omegaval)$ with
$\omegabw = \{ \parvar \mapsto a, \parconst_0 \mapsto a, \parconst_1 \mapsto a \}$
or $\omegabw = \{ \parvar \mapsto 3, \parconst_0 \mapsto 3, \parconst_1 \mapsto 3 \}$.
It is
not well sorted when $\omegabw = \{ \parvar \mapsto a_1, \parconst_0 \mapsto a_1, \parconst_1 \mapsto a_2 \}$
since $\instpar{\varphi}{\omega}{\I}$
is not a well sorted $\sigfsbv$-formula whenever $a_1^\I \not= a_2^\I$.
Note that an $\omega$ where $\omegabw(\parvar) = a_1 - a_2$
is not admissible, since $(a_1 - a_2)^\I \leq 0$ is possible even when $a_1^\I > 0$ and $a_2^\I > 0$.
\end{example}

Notice that symbolic constants
such as the maximum unsigned constant
of a symbolic length $w$ can be represented
by introducing $\parconst \in \parconsts$ with
$\omegabw(\parconst) = w$ 
and $\omegaval(\parconst) = 2^w-1$.
Furthermore, recall that signature $\sigfsbv$ includes
the bit-vector extract operator, which is parameterized 
by two natural numbers $u$ and $l$.
We do not lift the above definitions to handle
extract operations having symbolic ranges, e.g., where $u$ and $l$ are
$\sigia$-terms.
This is for simplicity and comes at no loss of expressive power,
since constraints involving extract can be equivalently expressed using
constraints involving concatenation.
For example, showing that every instance of a constraint
$s \teq \bvextract{t}{u}{l}$ holds,
where $0 < l \leq u < n-1$, is equivalent to showing that 
$t \teq \bvconcat{ y_1 }{ (\bvconcat{ y_2 }{ y_3 }) } \Rightarrow
s \teq y_2$
holds for all $y_1, y_2, y_3$,
where $y_1, y_2, y_3$ have sorts $\sortbv{n-1-u}$, $\sortbv{u-l+1}$,
$\sortbv{l}$, respectively.
We may reason about 
a formula involving a symbolic range $\{l, \ldots, u\}$ of $t$
by considering a parametric bit-vector formula that encodes
a formula of the latter form, where the appropriate symbolic bit-widths
are assigned to symbols introduced for $y_1, y_2, y_3$.

We assume the above definitions for 
parametric $\sigfsbv$-formulas are applied to parametric
$\sigfsbv$-terms as well.
Furthermore, for any admissible $\omega$,
we assume $\omega$ can be extended to terms $t$ of bit-vector sort
that are well sorted under $\omega$
such that
$\instpar{t}{\omega}{\I}$
has sort
$\sortbv{\omegabw(t)^{\I}}$
for all $\I$ that map variables in $\fv{\omega}$ to positive integers.
Such an extension of $\omega$ to terms can be easily computed
in a bottom-up fashion by computing $\omega$ for each child and then applying
the typing rules of the operators in $\sigfsbv$.
For example, we may assume $\omegabw(t) = \omegabw(t_2)$
if $t$ is of the form $\bvadd{t_1}{t_2}$ and is well sorted under $\omega$,
and $\omegabw(t) = \omegabw(t_1) + \omegabw(t_2)$
if $t$ is of the form $\bvconcat{t_1}{t_2}$.


Finally, we extend the notion of validity to parametric bit-vector
formulas. 
Given a formula $\varphi$ that is well sorted under $\omega$,
we say $\varphi$ is $\tfsbv$-valid under $\omega$
if $\instpar{\varphi}{\omega}{\I}$
is $\tfsbv$-valid 
for all $\I$ that that map variables in $\fv{\omega}$ to positive integers.


\section{Encoding Parametric Bit-Vector Formulas in SMT}
\label{sec:encodings}

Current SMT solvers do not support reasoning about
parametric bit-vector formulas.
In this section, we present a technique for encoding
such formulas
as formulas involving non-linear integer arithmetic,
uninterpreted functions, and universal quantifiers.
In SMT parlance, these are formulas in the \logicufnia logic.
%
%
Given a formula~$\varphi$
that is well sorted under some mapping $\omega$,
we describe this encoding in terms of a translation \traf,
which returns a formula $\psi$ that
is valid in the theory of
uninterpreted functions with integer arithmetic only if
$\varphi$ is $\tfsbv$-valid under $\omega$.
We describe several variations on this translation
and discuss their relative strengths and weaknesses.
\medskip

\noindent\textbf{Overall Approach} \ 
At a high level, our translation
produces an implication whose antecedent requires the integer variables 
to be in the correct ranges (e.g., $k>0$ for every bit-width variable $k$),
and whose conclusion is the result of 
converting each (parametric) bit-vector term of bit-width $k$
to an integer term.
Operations on parametric bit-vector terms are converted to operations
on the integers modulo $2^k$, 
where $k$ can be a symbolic constant.
We first introduce uninterpreted functions that 
will be used in our translation.
Note that SMT solvers may not support the full set of functions in our
extended signature \sigia, since they typically do not support exponentiation.
Since translation requires a limited form of exponentiation
%
we introduce an uninterpreted function symbol $\twotothef$
of sort $\Inte \to \Inte$,
whose intended semantics is the function $\lambda x.2^x$ 
when the argument $x$ is non-negative.
Second,
for each (non-predicate) $n$-ary (with $n > 0$) function $\bvsymbol{f}\hskip-.3em$ of sort 
$\sigma_1 \times \ldots\times \sigma_n \to \sigma$
in the signature of fixed-size bit-vectors $\sigfsbv$
(excluding bit-vector extraction),
we introduce an uninterpreted function
$\natsymbol{f}$ of arity $n+1$ and sort
$\Inte \times \Inte \times \ldots \times \Inte \to \Inte$, where the extra
argument is used to specify the bit-width.
For example, for $\bvaddf$ with sort 
$\sortbv{n} \times \sortbv{n} \to \sortbv{n}$,
we introduce
$\intaddf$ of sort
$\Inte \times \Inte \times \Inte \to \Inte$.
In its intended semantics, 
this function adds the second and third arguments, both integers,
and returns the result modulo $2^k$, where $k$ is the first argument.
The signature $\sigfsbv$ contains one function,
bit-vector concatenation $\bvconcatf$, whose
two arguments may have different sorts.
For this case,
the first argument
of $\intconcatf$
indicates the bit-width of the third argument, i.e.,
$\intconcat{k}{x}{y}$ is interpreted
as the concatenation of $x$ and $y$,
where $y$ is an integer that encodes a bit-vector of bit-width $k$;
the bit-width for $x$ is not specified by an argument,
as it is not needed for the elimination of this operator we perform later.
We introduce uninterpreted functions for each bit-vector predicate
in a similar fashion.
For instance,
$\intugef$ has sort $\Inte \times \Inte \times \Inte \to \bool$
and encodes whether its second argument
is greater than or equal to its third argument,
when these
two arguments are interpreted as unsigned bit-vector values
whose bit-width is given by its first argument.
Depending on the variation of the encoding,
our translation will either introduce quantified formulas that 
fully axiomatize the behavior of these uninterpreted
functions or add (quantified) lemmas that state key properties about them,
or both.
\medskip

\begin{figure}[t]
\begin{mdframed}
\noindent
$\traf_A$($\varphi$, $\omega$):\\[-4ex]
\begin{enumerate}
\item[] Return $\aaxiom{A}(\varphi,\omega) \Rightarrow \runconvert(\varphi,\omega)$.
\end{enumerate}
\vspace{-.5em}
\runconvert($e$, $\omega$):\\[-4ex]
\begin{enumerate}
\item[] Match $e$:
\item[] \begin{tabular}{lcll}
$\parvar$ & $\to$ & $\maptointvar(\parvar)$ & if $\parvar \in \parvars$ \\
$\parconst$ & $\to$ & $\omegaval(\parconst) \mod \twotothef(\omegabw(\parconst))$ & if $\parconst \in \parconsts$ \\
$t_1 \teq t_2$ & $\to$ & $\runconvert(t_1, \omega) \teq \runconvert(t_2, \omega)$ \\
$\bvsymbol{f}( t_1, \ldots, t_n )$ & $\to$ & $\runelim(\natsymbol{f}( \omegabw(t_n), \runconvert(t_1, \omega), \ldots, \runconvert(t_n, \omega) ))$ \\
${\bowtie}( \varphi_1, \ldots, \varphi_n )$ & $\to$ & ${\bowtie}(\runconvert(\varphi_1, \omega), \ldots, \runconvert(\varphi_n, \omega))$ & \hspace{-5em}${\bowtie}\, \in \{ \w,\v,\Ra,\neg,\Leftrightarrow \}$\\ 
\end{tabular}
\end{enumerate}
\vspace{-.3em}
\runelim($e$):\\[-4ex]
\begin{enumerate}
\item[] Match $e$:
\item[] \begin{tabular}{lcll}
$\intadd{k}{x}{y}$ & $\to$ & $(x+y) \mod \twotothe{k}$ \\
$\intsub{k}{x}{y}$ & $\to$ & $(x-y) \mod \twotothe{k}$\\
$\intmul{k}{x}{y}$ & $\to$ &  $(x\cdot y) \mod \twotothe{k}$\\
$\intudiv{k}{x}{y}$ & $\to$ &  $\ite{\equal{y}{0}}{ \twotothe{k}-1}{ x \div y }$ \\
$\inturem{k}{x}{y}$ & $\to$ &  $\ite{\equal{y}{0}}{ \twotothe{k}-1}{ x \mod y }$ \\
$\intnot{k}{x}$ & $\to$ & $\twotothe{k} - (x+1)$\\
  $\intneg{k}{x}$ & $\to$ & $(\twotothe{k} - x) \mod \twotothe{k}$\\
$\intshl{k}{x}{y}$ & $\to$ & $(x\cdot\twotothe{y}) \mod \twotothe{k}$\\
$\intlshr{k}{x}{y}$ & $\to$ &$(x\div\twotothe{y}) \mod \twotothe{k}$\\
$\intconcatf(k, x, y )$ & $\to$ & $x \cdot \twotothe{k} + y$\\
$\natsymbol{\bowtie}_u(k,x,y)$ & $\to$ & $x \bowtie y$ & $\bowtie\, \in \{ <, \leq, >, \geq \}$ \\
$\natsymbol{\bowtie}_s(k,x,y)$ & $\to$ & $\uts{k}{x} \bowtie \uts{k}{y}$ & $\bowtie\, \in \{ <, \leq, >, \geq \}$ \\
$e$ & $\to$ & $e$ & otherwise
\vspace{-.5em}
\end{tabular}
\end{enumerate}
\end{mdframed}
  \vspace{-2ex}
  \caption{Translation $\traf_A$ for parametric bit-vector formulas,
  parametrized by axiomatization mode~$A$.
We use 
$\uts{k}{x}$ as shorthand for
$2\cdot(x \mod \twotothe{k-1})- x$.
  }
  \label{fig:translation}
  \vspace*{-3ex}
\end{figure}

\noindent\textbf{Translation Function} \ 
Figure~\ref{fig:translation}
defines our translation function $\traf_A$,
which is parameterized by an axiomatization mode $A$.
Given an input formula $\varphi$
that is well sorted under $\omega$,
it returns the implication whose 
antecedant is an \emph{axiomatization} formula $\aaxiom{A}(\varphi,\sigma)$
and whose conclusion is
the result of converting $\varphi$ to its encoded version
via the conversion function $\runconvert$.
The former is dependent 
upon the axiomatization mode $A$ which we discuss later.
We assume without loss of generality that $\varphi$ contains
no applications of bit-vector extract, which can be eliminated
as described in the previous section, nor does it
contain concrete bit-vector constants, since these can be
equivalently represented by introducing a symbol in $\parconsts$
with the appropriate concrete mappings in $\omegabw$ and $\omegaval$.

In the translation, 
we use an auxiliary function $\runconvert$
which converts parametric bit-vector expressions
into integer expressions with uninterpreted functions.
Parametric bit-vector variables $\parvar$
(that is, symbols from $\parvars$) are replaced
by unique integer variables of type $\Inte$,
where we assume a mapping $\maptointvar$ maintains
this correspondence,
such that range of $\maptointvar$ does not include any variable 
that occurs in  $\fv{\omega}$.
Parametric bit-vector constants
$\parconst$
(that is, symbols from set $\parconsts$) are replaced 
by the term 
$\omegaval(\parconst) \mod \twotothef(\omegabw(\parconst))$.
The ranges of the maps in $\omega$
may contain arbitrary $\sigia$-terms.
In practice, 
our translation handles only cases where these terms
contain symbols 
supported
by the SMT solver, as well as terms of the form $2^t$, which we assume
are replaced by $\twotothe{t}$ during this translation.
For instance, if $\omegabw(\parconst) = w+v$ and
$\omegaval(\parconst) = 2^w-1$,
then $\runconvert(z)$ returns $(\twotothe{w}-1) \mod \twotothe{w+v}$.
Equalities are processed by recursively running the 
translation on both sides.
The next case handles 
symbols from the signature
$\sigfsbv$, 
where symbols $\bvsymbol{f}$ are replaced
with the corresponding uninterpreted function $\natsymbol{f}$.
We take as the first argument $\omegabw(t_n)$,
indicating the symbolic bit-width of the last argument of $e$,
and recursively call $\runconvert$ on $t_1, \ldots, t_n$.
In all cases, $\omegabw(t_n)$ corresponds 
to the bit-width that the uninterpreted function 
$\natsymbol{f}$ expects based on its intended semantics
(the bit-width of the second argument for bit-vector concatenation,
or of an arbitrary argument for all other functions and predicates).
Finally, if the top symbol of $e$ is a Boolean connective 
we apply the conversion function recursively to all its children.

We run $\runelim$ for all applications of
uninterpreted functions $\natsymbol{f}$ introduced during the conversion,
which eliminates functions that correspond
 to a majority of the bit-vector operators.
These functions can be equivalently
expressed using integer arithmetic 
and $\twotothef$.
The ternary addition operation $\intaddf$,
that represents addition of two bit-vectors with 
their width $k$ specified as the first argument, is translated to
integer addition modulo $\twotothe{k}$.
Similar considerations are applied for 
$\intsubf$ and
$\intmulf$.
For $\intudivf$ and $\inturemf$,
our translation handles the special case where the second argument is zero,
where the return value in this case is the maximum value for the given
bit-width, i.e. $\twotothe{k}-1$.
The integer operators corresponding to unary (arithmetic) negation
and bit-wise negation can be eliminated in a straightforward way.
The semantics of various bitwise shift operators can be defined arithmetically
using division and multiplication with $\twotothe{k}$.
Concatenation can be eliminated by multiplying its first argument $x$
by $\twotothe{k}$, where recall $k$ is the bit-width of the second arugment $y$.
In other words, it has the effect of shifting $x$ left by $k$ bits, 
as expected.
The unsigned relation symbols can be directly converted to the corresponding
integer relation.
For the elimination of signed relation symbols
we use an auxiliary helper $\utsf$ (unsigned to signed), 
defined in \Cref{fig:translation},
which returns the interpretation of its argument
when seen as a signed value.
The definition of $\utsf$ can be derived based on the semantics
for signed and unsigned bit-vector values in the SMT LIB standard.
Based on this definition,
we
have that integers $v$ and $u$ that encode bit-vectors
of bit-width $k$ satisfy $\intslt{k}{u}{v}$ if and only if $\uts{k}{u}<\uts{k}{v}$.

As an example of our translation,
let $\varphi=\bvadd{(\bvadd{\parvar}{\parvar})}{\parconst_1} \tneq \parconst_0$,
$\omegaval(\parconst_0)=0$, $\omegaval(\parconst_1)=1$, and
$\omegabw(\parvar)=\omegabw(\parconst_0)=\omegabw(\parconst_1)=a$
from \Cref{ex:parex}.
$\runconvert(\varphi, (\omegabw,\omegaval))$ is
$\runelim(\intadd{a}{\runelim(\intadd{a}{\maptointvar(\parvar)}{\maptointvar(\parvar)})}{1\mod \twotothe{a}})\tneq 0 \mod \twotothe{a}$.
After applying $\runelim$ and simplifying, we get
$(\maptointvar(\parvar) + \maptointvar(\parvar) + 1) \mod \twotothe{a} \tneq 0$.

Thanks to $\runelim$,
we can assume that all formulas generated by $\runconvert$
contain only uninterpreted function symbols in the set 
$\{ \twotothef,\intandf,\intorf,\intxorf \}$.
Thus, we restrict our attention to these symbols only 
in our axiomatization $\aaxiom{A}$, described next.
\bigskip

\begin{table}[t]
\centering
\begin{tabular}{|c|l|}\hline
\op&\axiom{\full}{\op}\\\hline
$\twotothef$ &
$
\equal{\twotothe{0}}{1}
\w
\forall k.\, k>0\Ra \equal{\twotothe{k}}{2\cdot\twotothe{k-1} }
$\\\hline
$\intandf$ & 
$
\begin{array} {l}
\forall k,x,y.\,
\intand{k}{x}{y} \teq \\
\quad 
\ite{k>1}{\intand{k-1}{x \mod \twotothe{k-1}}{y \mod \twotothe{k-1}}}{0} + {}\\
\quad 
\twotothe{k-1}\cdot
\min(\hselect{k}{k-1}{x},
\hselect{k}{k-1}{y}
)
\end{array}
$\\\hline
$\intxorf$ & 
$
\begin{array} {l}
\forall k,x,y.\,
\intxor{k}{x}{y} \teq\\
\quad 
\ite{k>1}{\intxor{k-1}{x \mod \twotothe{k-1}}{y \mod \twotothe{k-1}}}{0} + {} \\
\quad 
\twotothe{k-1}\cdot
\abs{\hselect{k}{k-1}{x}-
\hselect{k}{k-1}{y}}
\end{array}
$\\\hline

\end{tabular}
  \vspace{1ex}
\caption{Full axiomatization of \twotothef, \intandf, and \intxorf.
The axiomatization of $\intorf$ is omitted, and is dual to that of $\intandf$.
  We use
$\hselect{k}{i}{x}$ for
$(x \div \twotothe{i}) \mod 2$.
}
\label{fullax}
  \vspace*{-4ex}
\end{table}

\begin{table}[t]
\centering
\begin{tabular}{|l|l|l|}\hline
\op & axiom &$\axiom{\partial}{\op}$\\\hline
\multirow{5}{*}{$\twotothef$} &
base cases & 
$\equal{\twotothe{0}}{1}\w\equal{\twotothe{1}}{2}\w\equal{\twotothe{2}}{4}\w\equal{\twotothe{3}}{8}$ \\
& weak monotonicity &
$\forall i \forall j .\, i\leq j \Ra \twotothe{i}\leq\twotothe{j}$\\
& strong monotonicity &
$\forall i \forall j .\, i < j \Ra \twotothe{i} < \twotothe{j}$ \\
& modularity &
$\forall i \forall j \forall x .\, (x\cdot \twotothe{i}) \mod \twotothe{j} \tneq 0 \Ra i < j$ \\
& never even & 
$\forall i \forall x .\,  \twotothe{i}-1\tneq 2\cdot x$\\
& always positive &
$\forall i .\, \twotothe{i}\geq 1$ \\
& div 0 &
$\forall i .\, i \div \twotothe{i}\teq 0$ \\\hline

\multirow{7}{*}{$\intandf$} &
base case &  $\forall x\forall y.\, \equal{\intand{1}{x}{y}}{\min(\hselect{1}{0}{x}, \hselect{1}{0}{y})}$ \\
& max & 
$\forall k\forall x  .\, \equal{\intand{k}{x}{\intmax{k}}}{x}$
\\
& min & 
	$\forall k\forall x  .\, \equal{\intand{k}{x}{0}}{0}$
\\
&idempotence &
$\forall k\forall x.\, \equal{\intand{k}{x}{x}}{x}$ \\
& contradiction & 
	$\forall k \forall x .\, \equal{\intand{k}{x}{\intnot{k}{x}}}{0}$ \\
& symmetry &
	$\forall k \forall x \forall y .\, \equal{\intand{k}{x}{y}}{\intand{k}{y}{x}}$ \\
& difference &
$\forall k \forall x \forall y \forall z .\, x\tneq y\Ra\intand{k}{x}{z}\tneq y\v\intand{k}{y}{z}\tneq x$
\\
& range & 
$\forall k \forall x \forall y .\,  0\leq \intand{k}{x}{y}\leq \min(x,y)$\\\hline


\multirow{5}{*}{$\intxorf$} &
	base case &  $\forall x\forall y.\, 
	\equal{\intxor{1}{x}{y}}{\ite{\equal{\hselect{1}{0}{x}}{\hselect{1}{0}{y}}}{0}{1}}$ \\
& zero & 
	$\forall k\forall x  .\, \equal{\intxor{k}{x}{x}}{0}$
\\
& one & 
	$\forall k\forall x  .\, \equal{\intxor{k}{x}{\intnot{k}{x}}}{\intmax{k}}$
\\
& symmetry &
	$\forall k \forall x \forall y .\, \equal{\intxor{k}{x}{y}}{\intxor{k}{y}{x}}$ \\
& range & 
$\forall k \forall x \forall y .\,  0\leq \intxor{k}{x}{y}\leq \intmax{k}$\\\hline
\end{tabular}
\vspace{1ex}
\caption{Partial axiomatization of \twotothef, \intandf, and \intxorf.
The axioms for $\intorf$ are omitted, and are dual to those for $\intandf$.
We use $\intmax{k}$ for $\twotothe{k}-1$.
}
\label{partialax}
  \vspace*{-2ex}
\end{table}

\noindent\textbf{Axiomatization Modes} \ 
We consider four different axiomatization modes $A$,
which we call 
\full, \partial, \combined, and \qf (quantifier-free).
For each of these axiomatizations,
we define $\aaxiom{A}(\varphi,\omega)$ 
as the conjunction:
\[
\bigwedge_{x \in \fv{\varphi}} \!\!\!\!0 \leq \maptointvar(x) < \twotothe{\omegabw(x)}
\wedge (\!\!\!\!\bigwedge_{w \in \fv{\omega}}\!\!\!\! w > 0)
\wedge \axiom{A}{\twotothef}
\wedge \axiom{A}{\intandf} 
\wedge \axiom{A}{\intorf}
\wedge \axiom{A}{\intxorf}
\]
The first conjunction states that all integer variables
introduced for parametric bit-vector variables $\parvar$
reside in the range specified by their bit-width.
The second conjunction states that all free variables in $\omega$
(denoting bit-widths) are positive.
The remaining four conjuncts
denote the axiomatizations for the four uninterpreted functions
that may occur in the output of the conversion function.
The definitions of these formulas
are given in Tables~\ref{fullax} and~\ref{partialax}
for \full and \partial respectively.
For each axiom, 
$i,j,k$ denote bit-widths and $x,y$ denote integers that encode
bit-vectors of size $k$.
We assume guards on all quantified formulas
(omitted for brevity) that constrain $i,j,k$ to be positive
and $x,y$ to be in the range $\{0, \ldots, \twotothe{k}-1\}$.
Each table entry lists a set of formulas (interpreted conjunctively)
that state properties about the intended semantics
of these operators.
The formulas for axiomatization mode \full assert
the intended semantics
of these operators, 
whereas those for \partial assert several properties of them.
Mode \combined asserts both, and mode \qf
takes only the formulas in \partial that are 
quantifier-free.
In particular, $\axiom{\qf}{\twotothef}$ corresponds to the base cases listed in \partial, and
$\axiom{\qf}{\op}$ for the other operators is simply $\true$.
The partial axiomatization of these operations mainly includes natural properties of them. 
For example, we include some base cases for each operation, and also the ranges of its inputs and output.
For some proofs, these are sufficient. 
For $\intandf$\ , $\intorf$ and $\intxorf$, we also included their behavior for specific cases, e.g., $\intand{k}{a}{0}=0$ and its variants. 
Other axioms (e.g., ``never even") were added after analyzing specific
benchmarks to identify sufficient axioms for their proofs.

Our translation satisfies the following key properties.

\begin{theorem}
\label{tra-correct}
Let $\varphi$ be a parameteric bit-vector formula that is well sorted under $\omega$
and has no occurrences of bit-vector extract or concrete bit-vector constants. Then:
\be
\item\label{full-item} $\varphi$ is $\tfsbv$-valid under $\omega$ if and only if $\trfull{\varphi, \omega}$ is $\tufia$-valid.
\item\label{combined-item} $\varphi$ is $\tfsbv$-valid under $\omega$ if and only if $\trcombined{\varphi, \omega}$ is $\tufia$-valid.
\item\label{partial-item} $\varphi$ is $\tfsbv$-valid under $\omega$ if $\trpartial{\varphi, \omega}$ is $\tufia$-valid.
\item\label{qf-item} $\varphi$ is $\tfsbv$-valid under $\omega$ if $\trqf{\varphi, \omega}$ is $\tufia$-valid.
\ee
\end{theorem}

The proof of Property 1 is carried out by translating 
every interpretation 
$\bvi{\I}$ 
of $\tfsbv$
into a corresponding interpretation 
$\ufiai{\I}$ 
of $\tufia$
such that
$\bvi{\I}$ satisfies $\varphi$ iff $\ufiai{\I}$ satisfies $\trfull{\varphi}$.
The converse translation can be achieved similarly, where appropriate bit-widths are determined by the range axioms
$0 \leq \maptointvar(x) < \twotothe{\omegabw(x)}$ that occur in $\trfull{\varphi,\omega}$.
The rest of the properties follow from Property 1, by showing that the axioms in \Cref{partialax}
are valid in every interpretation of $\tufia$ that satisfies 
$\aaxiom{\full}(\varphi,\omega)$.


\section{Case Studies}
\label{sec:casestudy}
We apply the techniques from \Cref{sec:encodings}
to three case studies:
$(i)$ verification of invertibility conditions from Niemetz et al.~\cite{invcav18};
$(ii)$ verification of compiler optimizations as generated by~\alive \cite{alive15};
and $(iii)$ verification of rewrite rules that are used in SMT solvers.
For these case studies,
we consider a set of verification conditions that
originally use fixed-size bit-vectors,
and exclude formulas involving multiple bit-widths.

For each formula $\phi$, we first extract a parametric version $\varphi$
by replacing each
variable in $\phi$ by a fresh $\parvar \in \parvars$
and each (concrete) bit-vector constant by a fresh $\parconst \in \parconsts$.
We define $\omegabw(\parvar) = \omegabw(\parconst) = k$
for a fresh integer variable $k$,
and let $\omegaval(\parconst)$ be the
integer value corresponding to the bit-vector constant it
replaced.
Notice that, although omitted from the presentation,
our translation can be easily extended to handle quantified bit-vector formulas,
which appear in some of the case studies.
We then define $\omega = ( \omegabw, \omegaval )$
and invoke our translation from Section~\ref{sec:encodings}
on the parametric bit-vector formula $\varphi$.
If the resulting formula is valid, the original verification
condition holds independent of the original bit-width.
In each case study,
we report on the success rates of determining the validity of these formulas
for axiomatization modes \full, \partial, \combined, and \qf.
Overall, axiomatization mode \combined yields the best results.
%

All experiments described below require tools with support for the SMT logic
\logicufnia.
We used all three participants in the \logicufnia division of the 2018
SMT competition: 
\cvcfour~\cite{CVC4} (GitHub master 6eb492f6),
\zthree~\cite{Z3} (version 4.8.4), and
Vampire~\cite{DBLP:conf/cav/KovacsV13} (GitHub master d0ea236).
\zthree and \cvcfour
use various strategies and techniques for quantifier instantiation
including E-matching~\cite{DBLP:conf/cade/MouraB07},
and enumerative~\cite{DBLP:conf/tacas/ReynoldsBF18} and 
conflict-based~\cite{DBLP:conf/fmcad/ReynoldsTM14} instantiation.
For non-linear integer arithmetic, \cvcfour
uses an approach
based on incremental linearization~\cite{DBLP:journals/tocl/CimattiGIRS18,DBLP:conf/sat/CimattiGIRS18,DBLP:conf/frocos/ReynoldsTJB17}.
Vampire is a superposition-based theorem prover for first-order logic 
based on the AVATAR framework~\cite{DBLP:conf/cav/Voronkov14},
which has been extended also to support some theories
including integer arithmetic~\cite{DBLP:conf/tacas/Reger0V18}.
%
%
%
We performed all experiments on a cluster with Intel Xeon E5-2637 CPUs
with 3.5GHz and 32GB of memory and
used a time limit of 300 seconds (wallclock) and a memory limit
of 4GB for each solver/benchmark pair.
We consider a bit-width independent property to be proved if at least one solver
proved it for at least one of the axiomatization modes.%
\footnote{%
All benchmarks, results, log files, and solver configurations are available at
\url{http://cvc4.cs.stanford.edu/papers/CADE2019-BVPROOF/}.
}

\subsection{Verifying Invertibility Conditions}

  Niemetz et al.~\cite{invcav18} present a technique for solving quantified bit-vector formulas
  that utilizes \emph{invertibility conditions} to generate symbolic
  instantiations.
  Intuitively,
  an invertibility condition~$\ic$ for a literal $\ell[x]$ is
  the exact condition under which $\ell[x]$ has a solution for $x$, i.e.,
  $\ic \lequiv \exists x.\ell[x]$.
  For example, consider bit-vector literal $\bvand{x}{s}\teq t$
  with $x \not\in \fv{s} \cup \fv{t}$;
  then, the invertibility condition for $x$ is $\bvand{t}{s}\teq{t}$.

  The authors define invertibility conditions for 
  a representative set of literals having a single occurrence of $x$,
  that involve the bit-vector operators listed in \Cref{tab:bvops}, excluding extraction, 
  as the invertibility condition for the latter is trivially $\true$.
  A considerable number of these conditions were determined by leveraging
  syntax-guided synthesis (\sygus)
   techniques~\cite{AlurBJMRSSSTU13}.
  The authors further verified the correctness of all conditions for
  bit-widths 1 to 65.
  However, a bit-width-independent formal proof of correctness of these
  conditions was left to future work.
  In the following, we apply the techniques of \Cref{sec:encodings}
  to tackle this problem.
  Note that for this case study, we exclude operators involving multiple
  bit-widths, namely bit-vector extraction and concatenation.
  For the former, all invertibility conditions are \true,
  and for the latter a hand-written proof of the correctness of its
  invertibility conditions can be achieved easily.

\bigskip

\noindent\textbf
{Proving Invertibility Conditions} \ 
Let $\ell[x]$ be a bit-vector literal of the form
$\diamond x \rel t$ or $x \diamond s \rel t$ (dually, $s \diamond x \rel t$)
with operators $\diamond$ and relations \rel as defined in \Cref{tab:bvops}.
To prove the correctness of an invertibility condition \ic for $x$
independent of the bit-width,
we have to prove the validity of the formula:
%
%
\begin{equation}
\ic\Leftrightarrow\exists x . \ell[x]
\label{eq:ic}
\end{equation}
where occurrences of $s$ and $t$ are implicitly universally quantified.
We then want to prove that \Cref{eq:ic} is $\tfsbv$-valid under $\omega$.
Considering the two directions of~\eqref{eq:ic} separately, we get:
\begin{equation}
\tag{rtl}
\exists x . \ell[x,s,t]\Rightarrow \ic[s,t]
\label{eq:icrtl}
\end{equation}
\begin{equation}
\tag{ltr}
 \ic[s,t]\Rightarrow\exists x . \ell[x,s,t]
\label{eq:icltr}
\end{equation}

\noindent
The validity of~\eqref{eq:icrtl} is equivalent to the unsatisfiability of
the quantifier-free formula:
\begin{equation}
\tag{rtl'}
\ell[x,s,t]\w \neg \ic[s,t]
\label{eq:icrtltag}
\end{equation}

\noindent
Eliminating the quantifier in \eqref{eq:icltr} is much trickier.
It typically amounts to finding a
symbolic value for~$x$ such that $\ell[x,s,t]$ holds provided that $\ic[s,t]$
holds.
We refer to such a symbolic value as a \emph{conditional inverse}.
\bigskip

\noindent\textbf{Conditional Inverses} \ 
Given an invertibility condition $\ic$ for $x$ in bit-vector literal $\ell[x]$,
we say that a term $\ci$ is a \emph{conditional inverse} for $x$
if $\ic\Ra\ell[\ci]$ is $\tfsbv$-valid.
For example, the term $s$ itself is a conditional inverse for $x$ in
the literal $\bvule{(\bvor{x}{s})}{t}$:
given that there exists some $x$ such that 
$\bvule{(\bvor{x}{s})}{t}$, we have that
$(\bvor{s}{s})\bvulef
t$.
When a conditional inverse $\ci$ for $x$ is found, we may replace \eqref{eq:icltr} by:
\begin{equation}
\tag{ltr'}
 \ic\Rightarrow \ell[\ci]
\label{eq:icltrtag}
\end{equation}

\noindent
 Clearly, \eqref{eq:icltrtag} implies \eqref{eq:icltr}.
 However, the converse may not hold, i.e., if \eqref{eq:icltrtag} is refuted,
 \eqref{eq:icltr} is not necessarily refuted.
 Notice that if the invertibility condition for $x$ is \true,
 the conditional inverse is in fact unconditional.
 The problem of finding a conditional inverse for a bit-vector literal
 $x \diamond s \rel t$ (dually, $s \diamond x \rel t$) can be defined as
 a \sygus problem by asking whether there exists a binary bit-vector function
 $C$ such that the (second-order) formula
 $\exists C\forall s\forall t. \ic \Ra C(s,t) \diamond s \rel t $
 is satisfiable.
%
  If such a function $C$ is found, then it is in fact a conditional inverse for
  $x$ in $\ell[x]$.
  %
  We synthesized conditional inverses for $x$ in $\ell[x]$ for bit-width 4
with variants of the grammars used in~\cite{invcav18}
  to synthesize invertibility conditions.
  For each grammar
  we generated 160 \sygus problems,
  one for each combination of bit-vector operator and relation from
  \Cref{tab:bvops}
  (excluding extraction and concatenation),
  counting commutative cases only once.
  We used the \sygus feature of the SMT solver
  \cvcfour~\cite{ReynoldsDKTB15} to solve these problems,
  and out of 160, we were able to synthesize candidate conditional inverses for 143
  invertibility conditions.
  %
  For 12 out of these 143, we found that the synthesized
  terms were not conditional inverses for every bit-width, by checking \eqref{eq:icltrtag} for
  bit-widths up to 64.

%
%
\bigskip

\noindent\textbf{Results} \ 
  %
  \Cref{tab:icresultsverification} provides detailed information on the
  results for the axiomatization modes \full, \partial, and \qf
  discussed in \Cref{sec:encodings}.
  We use \ltr and~\rtl to indicate that only direction left-to-right
  (\ref{eq:icltr} or~\ref{eq:icltrtag}) or right-to-left (\ref{eq:icrtltag}), respectively,
  were proved,
  and \both and~\none to indicate that both or~none, respectively, of the directions
  were proved.
  Additionally, we use \ltrci (resp.~\ltrnoci) to indicate that for
  direction left-to-right, formula~\eqref{eq:icltrtag}
  (resp.~\eqref{eq:icltr}) was proved with (resp.~without) plugging in a
  conditional inverse.

  Overall, out of 160 invertibility conditions, we were able to fully prove~110,
  and for 19 (17) conditions we were able to prove only direction
  \ref{eq:icrtltag} (\ref{eq:icltrtag}).
  For direction right-to-left,
  129 formulas~\eqref{eq:icrtltag} overall were successfully proved to be unsatisfiable.
  Out of these 129,
  32 formulas were actually trivial since the invertibility condition \ic was~\true.
  For direction left-to-right,
  overall, 127 formulas were proved successfully,
  and out of these, 102 (94) were proved using
  (resp.~not using) a conditional inverse.
  Furthermore, 33 formulas could only be proved when using a conditional inverse.
  %
  Thus, using conditional inverses was helpful for proving
  the correctness of invertibility conditions.

  Considering the different axiomatization modes,
  overall, 
  with 104 fully proved and only 17 unproved instances,
  \combined performed best.
  Interestingly,
even though axiomatization \qf only includes 
  some of the
  base cases of axiomatization \partial,
it still performs well.
  This may be due to the fact that in many cases,
  the correctness of the invertibility condition does not rely on any
  particular property of the operators involved.
  For example, the invertibility condition \ic for literal
  $\equal{\bvand{x}{s}}{t}$ is  $\equal{\bvand{t}{s}}{t}$.
  Proving the correctness of \ic amounts to coming up with the right
  substitution for $x$, without relying on any particular axiomatization of
  $\intandf$.
  In contrast, the invertibility condition \ic for literal $\dist{\bvand{x}{s}}{t}$ is
  $\dist{t}{0}\v\dist{s}{0}$.
  Proving the correctness of \ic 
  relies on axioms regarding $\bvandf$ and $\bvnotf$.
  Specifically, we have found that from \partial, it suffices to keep ``min" and ``idempotence" to prove \ic.
  Overall,
  from the 2696 problems that this case study included,
  \cvcfour proved $50.3$\%,
  \vampire proved $31.4$\%,
  and
  \zthree proved $33.8$\%,
  while $23.5$\% of the problems were proved by all solvers.


\begin{table}[t]
  \centering
{%
  \renewcommand{\arraystretch}{1.1}%
  \begin{tabular}{r@{\hspace{3.0em}}r@{\hspace{2.0em}}r@{\hspace{2.0em}}r@{\hspace{2.0em}}r@{\hspace{3.0em}}r@{\hspace{2.0em}}r@{\hspace{2.0em}}}
  \hline
    \textbf{Axiomatization} & \both & \rtl & \ltr & \none & \ltrci & \ltrnoci \\
    \hline
    \full & 64 & 18 & 22 & 56 & 72 & 51\\
    \partial & 76 & 14 & 26 & 44 & 78 & 81\\
    \qf & 40 & 22 & 22 & 76 & 50 & 51\\
    \combined & 104 & 21 & 18 & 17 & 99 & 79\\
    \hline
    Total (160) & 110 & 19 & 17 & 14 & 102 & 94\\
  \end{tabular}%
}

  \vspace{1ex}
  \caption{
  Invertibility condition verification using axiomatization modes
  \combined, \full, \partial, and \qf.
  Column \ltrci (\ltrnoci) counts left-to-right proved with (without)
  conditional inverse.
  }
  \label{tab:icresultsverification}
  \vspace*{-4ex}
\end{table}

\subsection{Verifying Alive Optimizations}
  Lopes et al.~\cite{alive15} introduces \alive,
  a tool for proving the correctness of compiler peephole optimizations.
  \alive has a high-level language for specifying optimizations.  The tool
  takes as input a description of an optimization in this high-level language
  and then automatically verifies that applying the optimization to an
  arbitrary piece of source code produces optimized target code that is
  equivalent under a given precondition.  It can also automatically translate verified optimizations into \cpp code that can be linked 
  into \llvm~\cite{DBLP:conf/cgo/LattnerA04}.
%
For each optimization, \alive generates four constraints that encode the
following properties, assuming that the precondition of the optimization holds:
  \begin{enumerate}
    \item \label{enum:alive_mem} \emph{Memory}\hskip .5em Source and Target
      yield the same state of memory after execution.
    \item \label{enum:alive_def} \emph{Definedness}\hskip .5em The target is
      well-defined
      whenever the source is.
    \item \label{enum:alive_poi} \emph{Poison}\hskip .5em The target produces
      so-called poison values (caused by \llvm's $\nsw$, $\nuw$, and $\exact$
      attributes) only when the source does.
    \item \label{enum:alive_eq} \emph{Equivalence} Source and target yield the
      same result after execution.
  \end{enumerate}

\noindent
From these verification tasks, \alive can generate benchmarks in \smtlib
format in the theory of fixed-size bit-vectors, with and without quantifiers.
For each task, types are instantiated with all possible valid type assignments
(for integer types up to a default bound of 64 bits).
%
In the following, we apply our techniques from \Cref{sec:encodings}
to prove \alive verification tasks independently from the bit-width.
For this, as in the Alive paper,
we consider the set of optimizations from the $\instcombine$
optimization pass of \llvm, provided as \alive translations
(433 total).\footnote{At \scriptsize\url{https://github.com/nunoplopes/alive/tree/master/tests/instcombine}}
Of these 433 optimizations, 113 are dependent on a specific bit-width; thus we focus on the remaining 320.
%
%
We further exclude optimizations that do not comply with the following
criteria:
\begin{itemize}
  \item In each generated \smtlib file, only a single bit-width is used. 
  \item All \smtlib files generated for a property (instantiated for all
    possible valid type assignments) must be identical modulo the bit-width (excluding, e.g.,  bit-width dependent constants other than $0$, $1$, (un)signed min/max, and the bit-width).
\end{itemize}
As a useful exception to the first criterion,
we included instances where all terms of bit-width 1 can be interpreted as
Boolean terms.
Overall, we consider bit-width independent verification conditions
\ref{enum:alive_mem}--\ref{enum:alive_eq}
for 180 out of 320 optimizations.
None of these
include memory operations or poison values,
and only some have definedness constraints (and those are simple).
Hence, the generated verification conditions
\ref{enum:alive_mem}--\ref{enum:alive_poi}
are trivial.
We thus only consider the equivalence verification
conditions for these 180 optimizations.
\bigskip

\noindent\textbf{Results} \  
  \Cref{tab:aliveresults} summarizes the results of verifying 
  the equivalence constraints for the selected 180 optimizations
  from the \instcombine LLVM optimization pass.
  It first lists all families, showing the number of bit-width independent
  optimizations per family (320 total).  The next column indicates how many in
  each family were in the set of 180 considered
  optimizations, and the remaining columns show how many of those considered
  were proved with each axiomatization mode.

\begin{table}[t]
  \centering
{%
  \renewcommand{\arraystretch}{1.1}%
  \begin{tabular}{r@{\hspace{2.0em}}r@{\hspace{3.5em}}r@{\hspace{1.5em}}r@{\hspace{1.5em}}r@{\hspace{1.5em}}r@{\hspace{1.5em}}r@{\hspace{1.5em}}r@{\hspace{2.5em}}}
  \hline
  \textbf{Family} & \textbf{Considered} & \multicolumn{4}{c}{\textbf{Proved}}\\
  && \full & \partial & \qf & \combined & Total \\
  \hline              
  AddSub (52) & 16 &	 	7 & 7 & 7 & 9 & 9 \\
  MulDivRem (29) & 5 &	 	1 & 2 & 1 & 3 & 3 \\
  AndOrXor (162) & 124 &	 	57 & 55 & 53 & 60 & 60 \\
  Select (51) & 26 & 		 	15 & 11 & 11 & 16 & 16 \\
  Shifts (17) & 9 & 		 	0 & 0 & 	0 & 0 & 0 \\
  LoadStoreAlloca (9) & 0 & 	0 & 0 & 0 & 0 & 0 \\
  \hline
  Total (320) & 180 & 	 	80 & 75 & 72 & 88 & 88 \\
  \end{tabular}%
}

  \vspace{1ex}
  \caption{
  \alive optimizations verification using axiomatizations
  \combined, \full, \partial and \qf.
  }
  \label{tab:aliveresults}
  \vspace*{-4ex}
\end{table}

  Overall, out of 180 equivalence verification conditions,
  we were able to prove 88.
  Our techniques were most successful for the AndOrXor family.
  This is not too surprising, since
  many verification conditions of this family
  require only Boolean reasoning and basic properties of ordering relations
  that are already included in the theory \tia.
  For example,
  given bit-vector term $a$ and bit-vector constants $C_1$ and $C_2$,
  optimization AndOrXor:979 essentially rewrites
  $(\bvslt{a}{C_1}\w\bvslt{a}{C_2})$ to $\bvslt{a}{C_1}$, 
  provided that 
  precondition
  $\bvslt{C_{1}}{C_{2}}$ holds.
  To prove its correctness, it suffices to apply the transitivity of
  $\bvsltf$ with Boolean reasoning.
  The same holds when lifting this equivalence to the integers, deducing the
  transitivity of $\intsltf$ from that of the builtin $<$ relation of~$\tia$.

  None of the 9 benchmarks from the Shifts family were proven.
  These benchmarks are more complicated than others.
  They combine bit-wise and arithmetical operations and thus rely on their
  axiomatization.
  Solving these benchmarks is an interesting challenge for future
  work.  Adding specialized axioms to \partial is one promising approach.

  Interestingly, for this case study,
  the results from the different axiomatization modes are very similar.
  This can again be explained by the fact that many optimizations rely
  on properties of the integers that are already included in \tia,
  without requiring any particular property of functions
  $\twotothef$, $\intandf{}$, $\intorf{}$ and $\intxorf$ (as in the above example).

  Note that we have also tried using our approach for proving the equivalence verification conditions 
  for up to a bit-width of 64. 
  However, all optimizations that were proven correct this way were already 
  proven correct for arbitrary bit-widths, 
  which suggests 
  that this restriction did not make the benchmarks easier.
  Overall,
  from the 720 problems in this case study,
  \cvcfour proved $42.6$\%,
  \vampire proved $36.2$\%,
  and
  \zthree proved $37.9$\%,
  while $32.5$\% of the problems were proved by all solvers.

\subsection{BV Rewriting}

SMT solvers for the theory of fixed-size bit-vectors heavily rely on rewriting
to reduce the size of the input formula prior to solving the problem.
Since these rewrite rules are usually implemented independently of the
bit-width, verifying that they hold for any bit-width is crucial for the
soundness of the solver.
For this case study,
we used a feature of the SyGuS solver in CVC4 that allows us to enumerate
equivalent bit-vector terms/formulas (rewrite candidates) for a certain
bit-width up to a certain term depth (nesting level of operators)~\cite{noetzli2019}.
We generated 1575 pairs of equivalent bit-vector terms of depth
three and 431
equivalent pairs of formulas of depth two for bit-width 4 and translated them
to integer problems with axiomatization modes \full, \partial,
\qf, and
\combined, resulting in $6300+1724 = 8024$ benchmarks in total.
Since rewrites that have been proved can be used to further axiomatize the
integer translation, we collected all proven rewrites after each run, added
them as axioms to the initial problems and reran the experiments.
This was repeated until we reached a fixpoint, i.e., no further rewrites
were proved.
%
%
With this approach, we were able to prove 409 out of the 435 formula
equivalences (94\%), reaching a fixpoint at the first iteration.
For the equivalent terms, we initially proved 878 out of the 1575 equivalences,
which increased to 935 (59\%) after adding all axioms from the first run,
reaching a fixpoint after two iterations.
Overall,
from the 8024 problems,
\cvcfour proved $64.2$\%,
\vampire proved $66.5$\%,
and
\zthree proved $64.2$\%,
while $63.8$\% of the problems were proved by all solvers.

%
%
%
%
%
%

\section{Conclusion and Further Research}
\label{sec:conc}
We have studied several translations from bit-vector formulas with parametric bit-width to the theories of integer arithmetic and uninterpreted functions.
The translations differ in the way that the operator $2^{(\_)}$ and bitwise logical operators are 
axiomatized, namely, fully (using induction) or partially (using some of their key properties).
Our empirical results show that state-of-the-art SMT solvers 
are capable of solving the 
translated formulas
for various benchmarks that originate
from the verification of invertibility conditions, \llvm optimizations,
and rewriting rules for fixed-size bit-vectors.

In future research,
we plan to investigate a translation of our results to
a proof assistant such as \coq,
for which a bit-vector library was recently developed~\cite{ekici2017smtcoq}.
This will involve
supporting proofs in the SMT solver
for non-linear arithmetic and quantifiers.
We believe that our promising experimental results with an integer encoding
indicate that this is a viable approach for automating bit-width independent proofs.
We also plan to explore satisfiable benchmarks, and to extend our approach for translating models.

%
%
%
 \bibliographystyle{splncs04}
 \bibliography{mybib}
\newpage
\begin{appendix}

\section{Verified Invertibility Conditions}
\Cref{tab:icresults} summaries the results of verifying invertibility conditions. For each invertibility condition, 
it states whether it was fully proved, only one direction of it was proved, or none of the directions were proved.

\begin{table}
  \centering
{%
  \renewcommand{\arraystretch}{1.1}%
  \begin{tabular}{r@{\hspace{2.0em}}c@{\hspace{1.0em}}c@{\hspace{1.5em}}c@{\hspace{1.0em}}c@{\hspace{1.0em}}c@{\hspace{1.0em}}c@{\hspace{1.5em}}c@{\hspace{1.0em}}c@{\hspace{1.0em}}c@{\hspace{1.0em}}c}
    \hline
    \\[-2.5ex]
    $\ell[x]$ & \teq & \tneq & \bvultf & \bvugtf & \bvulef & \bvugef & \bvsltf & \bvsgtf & \bvslef & \bvsgef
    \\[.5ex]
    \hline
    \\[-2.5ex]
    $\bvneg{x}     \rel t$ & \both & \both & \both & \both & \both & \both & \both & \both & \both & \both\\
    $\bvnot{x}     \rel t$ & \both & \both & \both & \both & \both & \both & \both & \both & \both & \both\\
    $\bvand{x}{s}  \rel t$ & \ltr & \both & \both & \both & \both & \both & \ltr & \ltr & \none & \ltr\\
    $\bvor{x}{s}   \rel t$ & \ltr & \both & \both & \both & \both & \both & \ltr & \none & \ltr & \none\\
    $\bvshl{x}{s}  \rel t$ & \ltr & \rtl & \both & \ltr & \both & \ltr & \ltr & \none & \rtl & \none\\
    $\bvshl{s}{x}  \rel t$ & \both & \both & \both & \both & \both & \both & \rtl & \both & \rtl & \both\\
    $\bvlshr{x}{s} \rel t$ & \both & \both & \both & \ltr & \both & \both & \both & \ltr & \both & \ltr\\
    $\bvlshr{s}{x} \rel t$ & \both & \both & \both & \both & \both & \both & \both & \both & \both & \both\\
    $\bvashr{x}{s} \rel t$ & \none & \both & \both & \both & \both & \both & \ltr & \both & \ltr & \both\\
    $\bvashr{s}{x} \rel t$ & \both & \both & \rtl & \rtl & \rtl & \rtl & \rtl & \none & \rtl & \both\\
    $\bvadd{x}{s}  \rel t$ & \both & \both & \both & \both & \both & \both & \both & \both & \both & \both\\
    $\bvmul{x}{s}  \rel t$ & \none & \rtl & \both & \none & \both & \none & \none & \none & \rtl & \none\\
    $\bvudiv{x}{s} \rel t$ & \both & \both & \both & \both & \both & \rtl & \both & \both & \both & \both\\
    $\bvudiv{s}{x} \rel t$ & \both & \rtl & \both & \both & \both & \both & \both & \rtl & \both & \rtl\\
    $\bvurem{x}{s} \rel t$ & \both & \both & \both & \both & \both & \both & \none & \both & \rtl & \both\\
    $\bvurem{s}{x} \rel t$ & \ltr & \both & \both & \both & \both & \both & \both & \rtl & \both & \rtl\\
  \end{tabular}%
}

  \vspace{2ex}
  \caption{
  Invertibility conditions verification.
  \both means was fully proved,
  \protect\ltr  means only left-to-right proved,
  \protect\rtl means only right-to-left proved,
  \none means not proved.
  }
  \label{tab:icresults}
\end{table}

\newpage
\section{Conditional Inverses}

\Cref{tab:iv_eq,tab:iv_bvsl,tab:iv_bvsg,tab:iv_bvul,tab:iv_bvug}
list all verified conditional inverses that we found.
Note that we omitted superscript $\fsbvtheory$ from all bit-vector symbols for
better readability.

\begin{table}
  \centering
  {\renewcommand{\arraystretch}{1.1}%
\begin{tabular}{r@{\hspace{1em}}r@{\hspace{1em}}r}
\hline
\textbf{Literal} & \teqf{}{} & \tneqf{}{} \\
\hline
\ensuremath{\bvnneg{x}\rel t} & \ensuremath{\bvnneg{t}} & \ensuremath{\bvnnot{t}}\\
\ensuremath{\bvnnot{x}\rel t} & \ensuremath{\bvnnot{t}} & \ensuremath{t}\\
\ensuremath{\bvnadd{x}{s}\rel t} & \ensuremath{\bvnsub{t}{s}} & \ensuremath{\bvnnot{(\bvnadd{s}{t})}}\\
\ensuremath{\bvnand{x}{s}\rel t} & \ensuremath{t} & \ensuremath{\bvnnot{t}}\\
\ensuremath{\bvnashr{x}{s}\rel t} &  & \ensuremath{\bvnnot{t}}\\
\ensuremath{\bvnashr{s}{x}\rel t} &  & \ensuremath{\bvnlshr{t}{\bvnsub{s}{t}}}\\
\ensuremath{\bvnlshr{x}{s}\rel t} & \ensuremath{\bvnshl{t}{s}} & \ensuremath{\bvnshl{\bvnminsf}{t}}\\
\ensuremath{\bvnlshr{s}{x}\rel t} &  & \ensuremath{\bvnneg{t}}\\
\ensuremath{\bvnmul{x}{s}\rel t} &  & \ensuremath{\bvnshl{\bvnmaxsf}{t}}\\
\ensuremath{\bvnor{x}{s}\rel t} & \ensuremath{t} & \ensuremath{\bvnnot{t}}\\
\ensuremath{\bvnshl{x}{s}\rel t} & \ensuremath{\bvnlshr{t}{s}} & \ensuremath{\bvnshl{\bvnmaxsf}{t}}\\
\ensuremath{\bvnshl{s}{x}\rel t} &  & \ensuremath{t}\\
\ensuremath{\bvnudiv{x}{s}\rel t} & \ensuremath{\bvnmul{s}{t}} & \ensuremath{\bvnlshr{s}{t}}\\
\ensuremath{\bvnudiv{s}{x}\rel t} &  & \ensuremath{\bvnand{t}{\bvnminsf}}\\
\ensuremath{\bvnurem{x}{s}\rel t} & \ensuremath{t} & \ensuremath{\bvnneg{\bvnnot{t}}}\\
\ensuremath{\bvnurem{s}{x}\rel t} & \ensuremath{\bvnsub{s}{t}} & \ensuremath{t}\\
\hline
\end{tabular}}

  \vspace{2ex}
  \caption{Conditional inverses for relations \teq and \tneq.}
  \label{tab:iv_eq}
\end{table}

\begin{table}
  \centering
  {\renewcommand{\arraystretch}{1.1}%
\begin{tabular}{r@{\hspace{1em}}r@{\hspace{1em}}r}
\hline
\textbf{Literal} & \bvnslt{}{} & \bvnsle{}{} \\
\hline
\ensuremath{\bvnneg{x}\rel t} & \ensuremath{\bvnminsf} & \ensuremath{\bvnminsf}\\
\ensuremath{\bvnnot{x}\rel t} & \ensuremath{\bvnmaxsf} & \ensuremath{\bvnmaxsf}\\
\ensuremath{\bvnadd{x}{s}\rel t} & \ensuremath{\bvnsub{\bvnminsf}{s}} & \ensuremath{\bvnsub{t}{s}}\\
\ensuremath{\bvnand{x}{s}\rel t} & \ensuremath{\bvnminsf} & \ensuremath{t}\\
\ensuremath{\bvnashr{x}{s}\rel t} & \ensuremath{\bvnminsf} & \ensuremath{\bvnminsf}\\
\ensuremath{\bvnashr{s}{x}\rel t} & \ensuremath{\bvnnot{(\bvnor{s}{\bvnmaxsf})}} & \ensuremath{\bvnnot{(\bvnor{s}{\bvnmaxsf})}}\\
\ensuremath{\bvnlshr{x}{s}\rel t} & \ensuremath{\bvnshl{\bvnminsf}{s}} & \ensuremath{t}\\
\ensuremath{\bvnlshr{s}{x}\rel t} & \ensuremath{\bvnnot{(\bvnor{s}{\bvnmaxsf})}} & \ensuremath{\bvnnot{(\bvnor{s}{\bvnmaxsf})}}\\
\ensuremath{\bvnmul{x}{s}\rel t} &  & \\
\ensuremath{\bvnor{x}{s}\rel t} & \ensuremath{\bvnminsf} & \ensuremath{\bvnminsf}\\
\ensuremath{\bvnshl{x}{s}\rel t} & \ensuremath{\bvnlshr{\bvnminsf}{s}} & \ensuremath{\bvnlshr{t}{s}}\\
\ensuremath{\bvnshl{s}{x}\rel t} &  & \\
\ensuremath{\bvnudiv{x}{s}\rel t} & \ensuremath{\bvnnot{\bvnneg{t}}} & \ensuremath{t}\\
\ensuremath{\bvnudiv{s}{x}\rel t} &  & \\
\ensuremath{\bvnurem{x}{s}\rel t} & \ensuremath{\bvnnot{(\bvnor{\bvnmaxsf}{\bvnneg{s}})}} & \ensuremath{\bvnand{t}{\bvnminsf}}\\
\ensuremath{\bvnurem{s}{x}\rel t} & \ensuremath{t} & \ensuremath{\bvnsub{s}{t}}\\
\hline
\end{tabular}}

  \vspace{2ex}
  \caption{Conditional inverses for relations \bvnsltf and \bvnslef.}
  \label{tab:iv_bvsl}
\end{table}

\begin{table}
  \centering
  {\renewcommand{\arraystretch}{1.1}%
\begin{tabular}{r@{\hspace{1em}}r@{\hspace{1em}}r}
\hline
\textbf{Literal} & \bvnsgt{}{} & \bvnsge{}{} \\
\hline
\ensuremath{\bvnneg{x}\rel t} & \ensuremath{\bvnnot{t}} & \ensuremath{\bvnneg{t}}\\
\ensuremath{\bvnnot{x}\rel t} & \ensuremath{\bvnminsf} & \ensuremath{\bvnminsf}\\
\ensuremath{\bvnadd{x}{s}\rel t} & \ensuremath{\bvnsub{\bvnmaxsf}{s}} & \ensuremath{\bvnsub{t}{s}}\\
\ensuremath{\bvnand{x}{s}\rel t} & \ensuremath{\bvnmaxsf} & \ensuremath{\bvnmaxsf}\\
\ensuremath{\bvnashr{x}{s}\rel t} & \ensuremath{\bvnmaxsf} & \ensuremath{\bvnmaxsf}\\
\ensuremath{\bvnashr{s}{x}\rel t} & \ensuremath{\bvnand{s}{\bvnminsf}} & \ensuremath{\bvnand{s}{\bvnminsf}}\\
\ensuremath{\bvnlshr{x}{s}\rel t} & \ensuremath{\bvnshl{\bvnmaxsf}{s}} & \ensuremath{\bvnshl{t}{s}}\\
\ensuremath{\bvnlshr{s}{x}\rel t} &  & \\
\ensuremath{\bvnmul{x}{s}\rel t} &  & \\
\ensuremath{\bvnor{x}{s}\rel t} & \ensuremath{\bvnmaxsf} & \ensuremath{t}\\
\ensuremath{\bvnshl{x}{s}\rel t} & \ensuremath{\bvnlshr{\bvnmaxsf}{s}} & \ensuremath{\bvnlshr{\bvnmaxsf}{s}}\\
\ensuremath{\bvnshl{s}{x}\rel t} &  & \\
\ensuremath{\bvnudiv{x}{s}\rel t} &  & \\
\ensuremath{\bvnudiv{s}{x}\rel t} &  & \\
\ensuremath{\bvnurem{x}{s}\rel t} & \ensuremath{\bvnneg{\bvnnot{t}}} & \ensuremath{t}\\
\ensuremath{\bvnurem{s}{x}\rel t} & \ensuremath{\bvnsub{(\bvnor{s}{\bvnminsf})}{(\bvnand{\bvnmaxsf}{\bvnsub{t}{\bvnmaxsf}})}} & \ensuremath{\bvnsub{(\bvnor{s}{\bvnminsf})}{(\bvnand{t}{\bvnmaxsf})}}\\
\hline
\end{tabular}}

  \vspace{2ex}
  \caption{Conditional inverses for relations \bvnsgtf and \bvnsgef.}
  \label{tab:iv_bvsg}
\end{table}

\begin{table}
  \centering
  {\renewcommand{\arraystretch}{1.1}%
\begin{tabular}{r@{\hspace{1em}}r@{\hspace{1em}}r}
\hline
\textbf{Literal} & \bvnult{}{} & \bvnule{}{} \\
\hline
\ensuremath{\bvnneg{x}\rel t} & \ensuremath{0} & \ensuremath{0}\\
\ensuremath{\bvnnot{x}\rel t} & \ensuremath{\bvnneg{t}} & \ensuremath{\bvnnot{t}}\\
\ensuremath{\bvnadd{x}{s}\rel t} & \ensuremath{\bvnneg{s}} & \ensuremath{\bvnneg{s}}\\
\ensuremath{\bvnand{x}{s}\rel t} & \ensuremath{0} & \ensuremath{t}\\
\ensuremath{\bvnashr{x}{s}\rel t} & \ensuremath{0} & \ensuremath{0}\\
\ensuremath{\bvnashr{s}{x}\rel t} & \ensuremath{\bvnnot{(\bvnor{s}{\bvnmaxsf})}} & \ensuremath{\bvnnot{(\bvnor{s}{\bvnmaxsf})}}\\
\ensuremath{\bvnlshr{x}{s}\rel t} & \ensuremath{s} & \ensuremath{s}\\
\ensuremath{\bvnlshr{s}{x}\rel t} & \ensuremath{s} & \ensuremath{s}\\
\ensuremath{\bvnmul{x}{s}\rel t} & \ensuremath{0} & \ensuremath{0}\\
\ensuremath{\bvnor{x}{s}\rel t} & \ensuremath{s} & \ensuremath{s}\\
\ensuremath{\bvnshl{x}{s}\rel t} & \ensuremath{0} & \ensuremath{0}\\
\ensuremath{\bvnshl{s}{x}\rel t} & \ensuremath{\bvnminsf} & \ensuremath{\bvnminsf}\\
\ensuremath{\bvnudiv{x}{s}\rel t} & \ensuremath{0} & \ensuremath{t}\\
\ensuremath{\bvnudiv{s}{x}\rel t} & \ensuremath{\bvnnot{0}} & \ensuremath{\bvnnot{0}}\\
\ensuremath{\bvnurem{x}{s}\rel t} & \ensuremath{s} & \ensuremath{s}\\
\ensuremath{\bvnurem{s}{x}\rel t} & \ensuremath{s} & \ensuremath{s}\\
\hline
\end{tabular}}

  \vspace{2ex}
  \caption{Conditional inverses for relations \bvnultf and \bvnulef.}
  \label{tab:iv_bvul}
\end{table}

\begin{table}
  \centering
  {\renewcommand{\arraystretch}{1.1}%
\begin{tabular}{r@{\hspace{1em}}r@{\hspace{1em}}r}
\hline
\textbf{Literal} & \bvnugt{}{} & \bvnuge{}{} \\
\hline
\ensuremath{\bvnneg{x}\rel t} & \ensuremath{\bvnnot{t}} & \ensuremath{\bvnneg{t}}\\
\ensuremath{\bvnnot{x}\rel t} & \ensuremath{0} & \ensuremath{0}\\
\ensuremath{\bvnadd{x}{s}\rel t} & \ensuremath{\bvnnot{s}} & \ensuremath{\bvnnot{s}}\\
\ensuremath{\bvnand{x}{s}\rel t} & \ensuremath{s} & \ensuremath{s}\\
\ensuremath{\bvnashr{x}{s}\rel t} & \ensuremath{\bvnnot{0}} & \ensuremath{\bvnnot{0}}\\
\ensuremath{\bvnashr{s}{x}\rel t} & \ensuremath{\bvnand{s}{\bvnminsf}} & \ensuremath{\bvnand{s}{\bvnminsf}}\\
\ensuremath{\bvnlshr{x}{s}\rel t} & \ensuremath{\bvnnot{s}} & \ensuremath{\bvnnot{s}}\\
\ensuremath{\bvnlshr{s}{x}\rel t} & \ensuremath{0} & \ensuremath{0}\\
\ensuremath{\bvnmul{x}{s}\rel t} &  & \\
\ensuremath{\bvnor{x}{s}\rel t} & \ensuremath{\bvnnot{s}} & \ensuremath{t}\\
\ensuremath{\bvnshl{x}{s}\rel t} & \ensuremath{\bvnnot{0}} & \ensuremath{\bvnnot{0}}\\
\ensuremath{\bvnshl{s}{x}\rel t} &  & \\
\ensuremath{\bvnudiv{x}{s}\rel t} & \ensuremath{\bvnnot{0}} & \\
\ensuremath{\bvnudiv{s}{x}\rel t} & \ensuremath{0} & \ensuremath{0}\\
\ensuremath{\bvnurem{x}{s}\rel t} & \ensuremath{\bvnnot{\bvnneg{s}}} & \ensuremath{t}\\
\ensuremath{\bvnurem{s}{x}\rel t} & \ensuremath{0} & \ensuremath{0}\\
\hline
\end{tabular}}

  \vspace{2ex}
  \caption{Conditional inverses for relations \bvnugtf and \bvnugef.}
  \label{tab:iv_bvug}
\end{table}

\section{Proof of \Cref{tra-correct}}

\subsection{Property 1}
Given a mapping $\J$ of the variables in $\vec{v}=\fv{\omega}$ to positive integers,
every interpretation
$\bvi{\I}$ 
of $\tfsbv$
is translated to a corresponding  interpretation 
$\ufiai{\I}$ 
of $\tufia$
such that 
\begin{quote}
$(\ast)~\bvi{\I}$ satisfies $\varphi$ iff $\ufiai{\I}$ satisfies $\trfull{\varphi}$.
\end{quote}
$\ufiai{\I}$ is defined as follows:
\bi
\item $\op^{\ufiai{\I}}$
is set to satisfy $\axiom{\op}{\full}$ for any
$\op\in\set{\twotothef,\intandf,\intorf,\intxorf}$
\item $v^{\ufiai{\I}}=v^{\J}$ for every $v\in\vec{v}$
\item $\maptointvar(x)^{\ufiai{\I}}=\tonat{x_{c}^{\bvi{\I}}}$,
where $c=\omegabw(x)^{\J}$
for any $x\in \parvars$
\ei
The converse translation can be achieved similarly, where appropriate bit-widths are determined by the range axioms
$0 \leq \maptointvar(x) < \twotothe{\omegabw(x)}$ that occur in $\trfull{\varphi,\omega}$.
We prove $(\ast)$ using the following lemmas.

\begin{lemma}
\label{fullpow}
Let $\I$ be an interpretation of $\tufia$ that satisfies
$\axiom{\full}{\twotothef} $.
Then, over natural numbers, $\twotothef^{\I}$ is identical to $\lambda x.2^{x}$.
\end{lemma}

\begin{proof}
By the usual inductive definition of exponentiation.
\end{proof}

\begin{lemma}
\label{concatlem}
Let $\I$ be an interpretation of $\tufia$ that satisfies
$\axiom{\full}{\twotothef} $.
Let $a$ be a bit-vector constant of bit-width $k$ and
 $n=\tonat{a}$. 
Then $\I$ satisfies the following equations:
\bi
\item $\tonat{\bvconcat{i}{a}}\teq\twotothe{k}\cdot i+\tonat{a}$ for any $0\leq i\leq 1$.
\item $n \mod \twotothe{k-1}\teq\tonat{\bvextract{a}{k-2}{0}}$
\item $n \div 2 \teq \tonat{\bvextract{a}{k-1}{1}}$
\ei

\end{lemma}

\begin{proof}
By the definition of $\tonat{\cdot}$ and \Cref{fullpow}.
\end{proof}

\begin{lemma}
\label{fullselect}
Let $\I$ be an interpretation of $\tufia$ that satisfies
$\axiom{\full}{\twotothef} $,
 $a$ be a bit-vector constant of bit-width $k$,
$0\leq i\leq k-1$
 and
 $n=\tonat{a}$.
Then
 ${\hselect{k}{i}{n}}^{\I}=\bvselect{a}{i}$, where $\hselect{k}{i}{n}$ is defined in \Cref{fullax}.

\end{lemma}

\begin{proof}
We prove the first item by induction on $k$. The second is proved similarly.
If $k=1$ then we must have $i=0$.
In this case,
$0\leq n\leq 1$ and
$\hselect{k}{i}{n}^{\I}=n=\tonat{\bvselect{a}{0}}$.
Suppose $k>1$.
If $i=0$ then this is shown similarly to the base case.
Otherwise,
$\bvselect{a}{i}$ is the $i-1$ bit of $\bvextract{a}{1}{k-1}$.
By the induction hypothesis and \Cref{concatlem}, the latter is equal to
$\hselect{k-1}{i-1}{\tonat{\bvextract{a}{1}{k-1}}}^{\I}=
\hselect{k-1}{i-1}{\tonat{a} \div 2}^{\I}=
\hselect{k}{i}{n}^{\I}
$.
\end{proof}

\begin{lemma}
\label{fulllem}
Let $a$ and $b$ be bit-vector constants of bit-width $k$,
 $\op\in\set{\&,\mid,\oplus}$, and
 $\I$ an interpretation of $\tufia$ that satisfies 
$
\axiom{\full}{\twotothef} 
\wedge \axiom{\full}{\op}
$.
Then $\I$ satisfies
$\equal{\tonat{a\bvsymbol{\op}b}}
{\natsymbol{\op}(k,\tonat{a},\tonat{b})}$.
\end{lemma}

\begin{proof}
We prove the lemma for the case where $\op$ is $\&$ by induction on $k$.
The other cases are shown similarly.
If $k=1$, then by \Cref{fullpow,fullselect} we have 
$\tonat{\bvand{a}{b}}^{\I}=\tonat{\min(a,b)}=\min(a,b)=\intand{1}{\tonat{a}}{\tonat{b}}$.
Now suppose $k>1$.
Then,
$$
\begin{array} {l}
\intand{k}{\tonat{a}}{\tonat{b}}^{\I} = \\
(\intand{k-1}{\tonat{a} \mod 
\twotothe{k-1}}{
\tonat{b} \mod 
\twotothe{k-1}}+ \\
\twotothe{k-1}\cdot
\min(\hselect{k}{k-1}{a},
\hselect{k}{k-1}{b}))^{\I}
\end{array}
$$

By \Cref{fullselect}, we have that
$$\min(\hselect{k}{k-1}{a},
\hselect{k}{k-1}{b})=
\tonat{\bvand{\bvselect{a}{k-1}}{\bvselect{b}{k-1}}}
$$
By the induction hypothesis and \Cref{concatlem},
 we obtain that
 $\intand{k}{\tonat{a}}{\tonat{b}}$ is equal in $\I$ to
$$
\tonat{\bvand{\bvextract{a}{k-2}{0}}{\bvextract{b}{k-2}{0}}}  + 
\twotothe{k-1}\cdot \tonat{\bvand{\bvselect{a}{k-1}}{\bvselect{b}{k-1}}}
$$
which
by \Cref{concatlem} is equal in $\I$ to 
$$\tonat{\bvconcat{\bvand{\bvselect{a}{k-1}}{\bvselect{b}{k-1}}}{\bvand{\bvextract{a}{k-2}{0}}{\bvextract{a}{k-2}{0}}}}=
\tonat{\bvand{a}{b}}$$

\end{proof}

\begin{lemma}
\label{interlem}
${\runconvert(t,\omega)}^{\ufiai{\I}}=\tonat{{\instpar{t}{\omega}{\J}}^{\bvi{\I}}}$
for any parametric $\sigfsbv$-term $t$.
\end{lemma}

\begin{proof}
First notice that 
for any $x\in\parvars\cup\parconsts$ we have that
$\omegabw(x)^{\ufiai{\I}}=\omegabw(x)^{\J}$
and
$\omegaval(x)^{\ufiai{\I}}=\omegaval(x)^{\J}$.
Using induction, the same holds for any parametric $\sigfsbv$-term $t$.

We prove the lemma by induction on $t$.
\bi
\item If $t$ is $x$ for some $x\in\parvars$: follows from the definition of $\ufiai{\I}$ for this case.
\item If $t$ is $z$ for some $z\in\parconsts$: 
$$
\begin{array}{l}
{\runconvert(t,\omega)}^{\ufiai{\I}}= \\
(\omegaval(z) \mod \twotothe{\omegabw(z)})^{\ufiai{\I}}= \\
\omegaval(z)^{\J} \mod \twotothe{\omegabw(z)^{\J}}
\end{array}
$$
Now,
$\instpar{z}{\omega}{\J}$ is the bit-vector constant of width
$k=\omegabw(z)^{\J}$ whose integer value
is
$\omegaval(z)^{\J}\mod 2^{k}$.
\item If $t$ is constructed from an operator other than $\bvandf$, $\bvorf$, $\bvxorf$, 
then this follows by the semantics of the various operators as it is defined 
in the \smtlib standard. 
We explicitly show that case where $t=\bvadd{t_{1}}{t_{2}}$. 
In this case,
$${\runconvert(t,\omega)}=
(\runconvert(t_{1},\omega) + \runconvert(t_{2}, \omega)) \mod \twotothe{\omegabw(t_{2})}$$ 
which by \Cref{fullpow} is interpreted 
in $\ufiai{\I}$ as
$$((\runconvert(t_{1},\omega))^{\ufiai{\I}} + (\runconvert(t_{2}, \omega))^{\ufiai{\I}}) \mod 2^k$$ for
$k=\omegabw(t_{2})^{\J}$.
By the induction hypothesis, the latter is equal to
$$(\tonat{{\instpar{t_{1}}{\omega}{\J}}^{\bvi{\I}}} +
\tonat{{\instpar{t_{2}}{\omega}{\J}}^{\bvi{\I}}})
\mod 
2^k
$$
which by the semantic of $\bvaddf$ as defined in the
\smtlib standard, is equal to 
$\tonat{(\instpar{(\bvadd{t_{1}}{t_{2}})}{\omega}{\J})^{\bvi{\I}}}$.
\item The operators $\bvandf$, $\bvorf$ and $\bvxorf$ rely on \Cref{fulllem}, rather than on the \smtlib standard. 
We explicitly show the case where 
$t=\bvand{t_{1}}{t_{2}}$. 
In this case,
$${\runconvert(t,\omega)}^{\ufiai{\I}}=
(\intandf(\omegabw(t_{2}),\runconvert(t_{1},\omega),\runconvert(t_{2}, \omega)))^{\ufiai{\I}}$$ 
which by the induction hypothesis is equal to 

$$
(\intandf(\omegabw(t_{2}),
\tonat{{\instpar{t_{1}}{\omega}{\J}}^{\bvi{\I}}},
\tonat{{\instpar{t_{2}}{\omega}{\J}}^{\bvi{\I}}}))^{\ufiai{\I}}
$$ 
By \Cref{fulllem}, we obtain 
$$
\tonat{(\instpar{(\bvand{t_{1}}{t_{1}})}{\omega}{\J})^{\bvi{\I}}}
$$
\ei
\end{proof}

\begin{corollary}
$\bvi{\I}$ satisfies $\varphi$ iff $\ufiai{\I}$ satisfies $\trfull{\varphi}$.
\end{corollary}

\begin{proof}
By routine induction on $\varphi$, where the base cases follow from \Cref{interlem}.
\end{proof}

\subsection{Properties 2 -- 4}
The rest of the properties follow from Property 1
 by showing that the axioms in \Cref{partialax}
are valid in every interpretation of $\tufia$ that satisfies 
$\aaxiom{\full}(\varphi,\omega)$.
The axioms of $\twotothef$ easily follow from simple properties of exponentiation.
As for the bitwise logical operators, 
we explicitly prove the validity of ``difference" for $\intandf$. The rest are shown similarly.
Let $k>0$, $0\leq x,y,z < 2^{k}$,
and $\I$ an interpretation of $\tufia$ that satisfies $\aaxiom{\full}(\varphi,\omega)$.
Note that $k$, $x$, $y$ and $z$ are bound as integers, and are therefore interpreted as themselves in $\I$.
Let $a$, $b$ and $c$ be bit-vectors of width $k$ such that
$x=\tonat{a}$, $y=\tonat{b}$ and $z=\tonat{c}$.
Suppose $\I$ satisfies $x\tneq y$. 
Then $a\tneq b$, and so there exists
 some $0\leq i\leq k-1$ such that
$\dist{\bvselect{a}{i}}{\bvselect{b}{i}}$.
First suppose $\bvselect{a}{i}=0$.
Then $\bvselect{b}{i}=1$ and $\bvselect{(\bvand{a}{c})}{i}=0$, and hence
$\dist{\bvand{a}{c}}{b}$.
This means that 
$\dist{\tonat{\bvand{a}{c}}}{\tonat{b}}$, and thus by
By \Cref{fulllem} we have that $\I$ satisfies
$\dist{{\intand{k}{x}{z}}}{y}$.
Otherwise, $\bvselect{a}{i}=1$ and $\bvselect{b}{i}=0$.
If $\equal{\bvand{a}{c}}{b}$ then
$\bvselect{c}{i}=0$, which means that 
$\bvselect{(\bvand{b}{c})}{i}=0$ and so
$\dist{\bvand{b}{c}}{a}$.
This means that 
$\dist{\tonat{\bvand{b}{c}}}{\tonat{a}}$, and thus by 
\Cref{fulllem} we have that $\I$ satisfies
$\dist{{\intand{k}{y}{z}}}{x}$.

\end{appendix}

\end{document}